\documentclass[a4paper,UKenglish,cleveref,autoref,thm-restate]{lipics-v2021}

\title{A structural operational semantics for interactions with a look at loops}

\titlerunning{A structural operational semantics for interactions}

\author{Erwan Mahe}{Laboratoire de Mathématiques et Informatique pour la Complexité et les Systèmes, CentraleSupélec, 9 rue Joliot-Curie, Gif-sur-Yvette, France}{}{https://orcid.org/0000-0002-5322-4337}{}

\author{Christophe Gaston}{Laboratory of Systems Requirements and Conformity Engineering, CEA-LIST, P.C. 174, Gif-sur-Yvette, France}{}{https://orcid.org/0000-0001-6865-5108}{}

\author{Pascale Le Gall}{Laboratoire de Mathématiques et Informatique pour la Complexité et les Systèmes, CentraleSupélec, 9 rue Joliot-Curie, Gif-sur-Yvette, France}{}{https://orcid.org/0000-0002-8955-6835}{}

\authorrunning{E.\,Mahe and C.\,Gaston and P.\,Le Gall}

\Copyright{Erwan Mahe and Chritophe Gaston and Pascale Le Gall} 



\category{}

\relatedversion{}

\nolinenumbers

\hideLIPIcs

\EventEditors{}
\EventNoEds{1}
\EventLongTitle{}
\EventShortTitle{}
\EventAcronym{}
\EventYear{}
\EventDate{}
\EventLocation{}
\EventLogo{}
\SeriesVolume{}
\ArticleNo{}

\usepackage{tikz}
\usepackage{bussproofs}
\usepackage{makecell}
\usepackage{mathtools}

\definecolor{bordeau}{rgb}{0.3515625,0,0.234375}
\definecolor{darkspringgreen}{rgb}{0.09, 0.45, 0.27}
\definecolor{CS_Red}{RGB}{150,2,60}
\definecolor{CS_LightRed}{RGB}{207,161,162}
\definecolor{CS_Grey}{RGB}{133,120,148}
\definecolor{hibou_col_lf}{RGB}{22, 22, 130}
\newcommand{\hlf}[1]{\textcolor{hibou_col_lf}{#1}}
\definecolor{hibou_col_pr}{RGB}{22, 130, 22}

\definecolor{hibou_col_ms}{RGB}{15, 86, 15}
\newcommand{\hms}[1]{\textcolor{hibou_col_ms}{#1}}

\newcommand{\shortColRed}[1]{\textcolor{red}{#1}}
\newcommand{\shortColBlue}[1]{\textcolor{blue}{#1}}
\newcommand{\shortColOrange}[1]{\textcolor{orange}{#1}}
\newcommand{\shortColViolet}[1]{\textcolor{violet}{#1}}

\usepackage{scalerel}

\DeclareRobustCommand\doubleVerticalTimesDefault{%
  \leavevmode
  {\sbox0{\ddag}%
   \ooalign{\raisebox{\ht0-\height}{$\times$}\cr
            \raisebox{\depth-\dp0}{\scalebox{1}[-1]{$\times$}}\cr}%
  }%
}

\newcommand{\doubleVerticalTimes}{\scalerel*{\doubleVerticalTimesDefault}{b}}

\newcommand{\globalInterleaving}{||}

\newcommand{\globalStrictSeq}{;}

\newcommand{\globalWeakSeq}{;_{\doubleVerticalTimes}}

\newcommand{\evadesLf}{\downarrow^{\doubleVerticalTimes}}

\newcommand{\collidesLf}{\not{\evadesLf}}

\newcommand{\isPruneBase}{\mathrlap{\raisebox{-.125\height}{\doubleVerticalTimes}}\longrightarrow}

\newcommand{\isPruneOf}[1]{\mathrlap{\raisebox{-.125\height}{\doubleVerticalTimes}}\xrightarrow{#1}}

\newtheorem*{lemma*}{Lemma}
\newtheorem*{theorem*}{Theorem}

\begin{document}

\maketitle

\begin{abstract}
Message Sequence Charts \& Sequence Diagrams are graphical models that represent the behavior of distributed and concurrent systems via the scheduling of discrete and local emission and reception events. We propose an Interaction Language (IL) to formalize such models, defined as a term algebra which includes strict and weak sequencing, alternative and parallel composition and four kinds of loops. This IL is equipped with a denotational-style semantics associating a set of traces (sequences of observed events) to each interaction.
We then define a structural operational semantics in the style of process algebras and formally prove the equivalence of both semantics.
\end{abstract}

\section{Introduction}

Specifying the behavior of distributed and concurrent systems is the object of a large literature. In particular modelling asynchronous communications between concurrent processes is possible under a variety of formalisms, including process algebras \cite{weak_sequential_composition_in_process_algebra}, Petri Nets \cite{dynamic_recursive_petri_nets}, series-parallel languages \cite{series_parallel_languages_and_the_bounded_width_property}, distributed automata \cite{distributed_timed_automata_with_independently_evolving_clocks}, or formalisms derived from Message Sequence Charts (MSC) \cite{operational_semantics_for_msc}. In the context of this paper, we focus on the latter. MSCs are graphical models which represent the exchange of information between sub-systems. Various offshoots of MSCs, which notably include UML Sequence Diagrams (UML-SD), have been proposed and we may call languages from that family "Interaction Languages" (IL). Interactions are particularly interesting due to their graphical nature and ease of understanding. 

In order to use interactions for formal verification, they have to be fitted with formal semantics. A major hurdle in defining those lies in the treatment of weak sequencing. In a few words, weak sequencing allows events taking place on different locations to occur in any order while strictly ordering those that take place on the same location.
The survey \cite{the_many_meanings_of_uml2_sd_a_survey} provides an overview of solutions found in the literature. 
Those can be roughly categorized as follows:
(a) direct denotational semantics, which define the trace semantics of interactions directly. They either rely on partial order sets \cite{semantics_of_interactions_in_uml_2_0,revisiting_semantics_of_interactions_for_trace_validity_analysis} or algebraic operators \cite{UML_interactions_meet_state_machines_an_institutional_approach}; (b) semantics by translation, which rely on translating interactions into some other models such as Petri Nets \cite{compositional_semantics_for_UML2_sequence_diagrams_using_petri_nets} or automata \cite{model_checking_of_uml2_interactions} and then use (generally operational-style) semantics of those other formalisms and; (c) direct operational-style semantics, which notably include that of \cite{operational_semantics_for_msc}.

While approaches in (c) are highly valuable for defining algorithms for formal verification, approaches in (a) are adapted to reason about interactions and their properties. Our contribution provides a framework encompassing (a) and (c) based on an IL which includes strict and weak sequencing, alternative and parallel composition and four semantically distinct loops. We propose a denotational-style semantics extending that of \cite{UML_interactions_meet_state_machines_an_institutional_approach} with the addition of repetition operators in the form of variants of the algebraic Kleene closure. We also define a structural operational semantics in the fashion of process algebras \cite{process_algebra_with_explicit_termination} with a particular care for the handling of weak sequencing and which uses dedicated rules to address the various loops. The equivalence of both semantics is formally proven (with an automated proof made with Coq available in \cite{coq_hibou_label_semantics_equivalence}).
Languages such as MSC or UML-SD only propose a single loop construct. By contrast, our contribution highlights the use of distinct loops to express nuances in the repetition of behaviors.

This paper is organized as follows: Sec.\ref{sec:basic_interactions} introduces the concept of interactions and traces. 
Sec.\ref{sec:semantic_domain} defines repetition operators on sets of traces. 
Sec.\ref{sec:syntax_denotational_semantics} presents the syntax of our IL and defines a trace semantics in denotational-style using operators introduced in the previous sections. 
Sec.\ref{sec:operational_semantics} defines a structural operational semantics on interactions with loops in the style of process calculus. In Sec.\ref{sec:proof_equivalence}, we demonstrate the equivalence of both semantics.
Finally, in Sec.\ref{sec:related_works} we discuss some related works and we conclude in Sec.\ref{sec:conclusion}. Detailed demonstrations are included in the appendix.

\section{Basic interactions \& intuition of their meaning\label{sec:basic_interactions}}

Interactions describe the behavior of distributed and concurrent systems through the perspective of their internal and external communications. They are defined up to a signature $\Omega = (L,M)$ where $L$ is a set of lifelines and $M$ is a set of messages.
Lifelines are abstractions of sub-systems, or groups of sub-systems while messages represent data sent and received by sub-systems.

\subsection{Preliminaries}

The executions of systems are characterized by sequences of events called communication actions (actions for short) which are of two kinds:
either the emission of a message $m \in M$ from a lifeline $l \in L$, denoted by $l!m$, or the reception of a message $m \in M$ by a lifeline $l \in L$, denoted by $l?m$.
We note $\mathbb{A}_\Omega$ the set of actions defined up to $\Omega$. For any such action $a$, $\theta(a)$ denotes the lifeline on which $a$ occurs.

Sequences of actions from $\mathbb{A}_\Omega$, called traces, are words in $\mathbb{A}_\Omega^*$, with "$.$" denoting the classical concatenation law and $\epsilon$ being the empty word (empty trace). We denote by $\mathbb{T}_\Omega = \mathbb{A}_\Omega^*$ the set of traces. Thus, for any two traces $t_1$ and $t_2$, $t_1.t_2$ is the trace composed of the sequence of actions of $t_1$ followed by the sequence of actions of $t_2$. 

We now introduce operators\footnote{We use notations from \cite{UML_interactions_meet_state_machines_an_institutional_approach}.} to compose traces. Those operators model different notions of scheduling: 
the interleaving ($\globalInterleaving$) and the weak sequencing ($\globalWeakSeq$).

The set $t_1 \globalInterleaving t_2$ of interleavings of traces $t_1$ and $t_2$ is defined by:
\[
\begin{array}{lcl}
\epsilon \globalInterleaving t_2 &=& \{ t_2 \}
\\
t_1 \globalInterleaving \epsilon &=& \{ t_1 \}
\\
(a_1.t_1) \globalInterleaving (a_2.t_2) &=& 
\{a_1.t ~|~ t \in t_1 \globalInterleaving (a_2.t_2)\}
\cup \{a_2.t ~|~ t \in (a_1.t_1) \globalInterleaving t_2\}
\end{array}
\]

Interleaving allows elements of distinct traces to be reordered w.r.t. one another while preserving the order that is specific to each trace.
By contrast, weak sequencing only allows such permutations when actions do not occur on the same lifeline. Defining weak sequencing requires introducing a predicate:
for any trace $t$ and lifeline $l$, the predicate $t \doubleVerticalTimes l$ means that $t$ contains at least an action occurring on $l$ (we say that $t$ has conflicts w.r.t. $l$). It is defined as follows:
\[
\begin{array}{lcl}
\epsilon \doubleVerticalTimes l &=& \bot\\
(a.t) \doubleVerticalTimes l &=& (\theta(a) = l) \vee (t \doubleVerticalTimes l)
\end{array}
\]

The set $t_1 \globalWeakSeq t_2$ of weak sequencing of traces $t_1$ and $t_2$ is defined by:
\[
\begin{array}{lcl}
\epsilon \globalWeakSeq t_2 &=& \{ t_2 \}
\\
t_1 \globalWeakSeq \epsilon &=& \{ t_1 \}
\\
(a_1.t_1) \globalWeakSeq (a_2.t_2) &=& 
\{ a_1.t ~|~ t \in t_1 \globalWeakSeq (a.t_2) \}
\cup 
\left\{ a_2.t ~\middle|~
\begin{array}{c}
t \in (a_1.t_1) \globalWeakSeq t_2,\; \neg (a_1.t_1 \doubleVerticalTimes \theta(a_2)) 
\end{array}
\right\}
\end{array}
\]

The previous binary operators ("$.$", "$\globalWeakSeq$" \& "$\globalInterleaving$") defined on traces are canonically extended to sets of traces as follows. 
Given $T_1$ and $T_2$ two sets of traces, $T_1 \globalStrictSeq T_2$ denotes the set of all traces $t_1.t_2$ s.t. $t_1 \in T_1$ and $t_2 \in T_2$. Likewise, with $\diamond \in \{\globalWeakSeq,~\globalInterleaving\}$, $T_1 \diamond T_2$ is the union of all the sets $t_1 \diamond t_2$ s.t. $t_1 \in T_1$ and $t_2 \in T_2$.
By choosing "$\globalStrictSeq$" for denoting the extension of "." to sets of traces, we adopt the same notation as in~\cite{UML_interactions_meet_state_machines_an_institutional_approach}. "$\globalStrictSeq$" is called the strict sequencing operator.
The use of the strict sequencing ("$;$"), weak sequencing ("$\globalWeakSeq$") and interleaving ("$\globalInterleaving$") operators is illustrated on the right of Fig.\ref{fig:basic_interaction_example} so as to compute the semantics of an example interaction. For instance, weak sequencing allows $\shortColRed{l_1!m_2}$ to be reordered before $\shortColOrange{l_3?m_1}$ but not before $\shortColOrange{l_1!m_1}$ while interleaving allows $\shortColBlue{l_1!m_4}$ to be placed anywhere w.r.t. $\shortColViolet{l_1!m_3}$ and $\shortColViolet{l_2?m_3}$.

\subsection{Basic interactions\label{subsec:basic_interactions}}

Interactions can be drawn as "sequence diagrams" of which an example is given in the left of Fig.\ref{fig:basic_interaction_example}. Lifelines are drawn as vertical lines. Emission and reception actions are drawn as horizontal arrows carrying the transmitted message and which respectively exit the emitting lifeline or point towards the receiving lifeline. When a direct emission-reception causality can be identified we draw both actions as a single arrow from the emitter towards the receiver.

The top to bottom direction relates to the passing of time. An action (arrow) drawn above another generally occurs beforehand. This scheduling of actions corresponds to the weak sequencing operator. By contrast, strict sequencing may be used to enforce precedence relations between actions occurring on different lifelines. In~Fig.\ref{fig:basic_interaction_example}, the arrow carrying $m_1$ and specifying its passing between $l_1$ and $l_3$ is modelled by the interaction $strict(\shortColOrange{l_1!m_1},\shortColOrange{l_3?m_1})$. The fact that this arrow stands above that carrying $m_2$ can be modelled using the weak sequencing operator: $seq(strict(\shortColOrange{l_1!m_1},\shortColOrange{l_3?m_1}),strict(\shortColRed{l_1!m_2},\shortColRed{l_2?m_2}))$. Using $seq$ here instead of $strict$ allows for instance $\shortColRed{l_2?m_2}$ to occur before $\shortColOrange{l_3?m_1}$ even though the latter is drawn above. However $\shortColRed{l_1!m_2}$ cannot occur before $\shortColOrange{l_1!m_1}$ because they both occur on~$l_1$.

Parallel and alternative compositions can also be used to specify more complex behavior. On Fig.\ref{fig:basic_interaction_example}, the passing of $m_3$ and the emission of $m_4$ are scheduled using parallel composition. In the diagram representation this corresponds to the box labelled with "par". This can be modelled with the term $par(strict(\shortColViolet{l_1!m_3},\shortColViolet{l_2?m_3}),\shortColBlue{l_1!m_4})$. Actions scheduled with $par$ can occur in any order w.r.t. one another. Here, $\shortColBlue{l_1!m_4}$ can occur before $\shortColViolet{l_1!m_3}$, after $\shortColViolet{l_2?m_3}$ or in between those two actions. Alternative composition proposes an exclusive non-deterministic choice between behaviors. Like $par$, $alt$ is drawn as a box labelled with "alt".

\begin{figure}
    \centering
\begingroup
\setlength{\tabcolsep}{1pt} 
\renewcommand{\arraystretch}{.2} 
\begin{tabular}{cc}
\makecell{
\includegraphics[scale=.325]{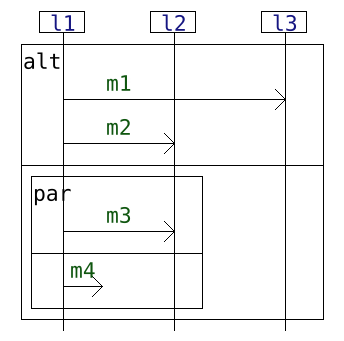}
}
&
\makecell[l]{
\[
\begin{array}{l}
= \left(
\begin{array}{c}
(\{\shortColOrange{l_1!m_1}\} \globalStrictSeq \{\shortColOrange{l_3?m_1}\})
\\
\globalWeakSeq 
(\{\shortColRed{l_1!m_2}\} \globalStrictSeq \{\shortColRed{l_2?m_2}\})
\end{array}
\right)
\cup 
\left(
\begin{array}{c}
(\{\shortColViolet{l_1!m_3}\} \globalStrictSeq \{\shortColViolet{l_2?m_3}\})
\\
\globalInterleaving
\{\shortColBlue{l_1!m_4}\}
\end{array}
\right)
\\
\\
= \left(
\begin{array}{c}
\{\shortColOrange{l_1!m_1}.\shortColOrange{l_3?m_1}\}
\\
\globalWeakSeq 
\{\shortColRed{l_1!m_2}.\shortColRed{l_2?m_2}\}
\end{array}
\right)
\cup 
\left(
\begin{array}{c}
\{\shortColViolet{l_1!m_3}.\shortColViolet{l_2?m_3}\}
\\
\globalInterleaving
\{\shortColBlue{l_1!m_4}\}
\end{array}
\right)
\\
\\
= \left\{
\begin{array}{c}
\shortColOrange{l_1!m_1}.\shortColOrange{l_3?m_1}.\shortColRed{l_1!m_2}.\shortColRed{l_2?m_2},\\
\shortColOrange{l_1!m_1}.\shortColRed{l_1!m_2}.\shortColOrange{l_3?m_1}.\shortColRed{l_2?m_2},\\
\shortColOrange{l_1!m_1}.\shortColRed{l_1!m_2}.\shortColRed{l_2?m_2}.\shortColOrange{l_3?m_1}
\end{array}
\right.
\left.
\begin{array}{c}
\shortColViolet{l_1!m_3}.\shortColViolet{l_2?m_3}.\shortColBlue{l_1!m_4},\\
\shortColViolet{l_1!m_3}.\shortColBlue{l_1!m_4}.\shortColViolet{l_2?m_3},\\
\shortColBlue{l_1!m_4}.\shortColViolet{l_1!m_3}.\shortColViolet{l_2?m_3}
\end{array}
\right\}
\\
\end{array}
\]
}
\end{tabular}
\endgroup 
    \caption{Example of a basic interaction \& its trace semantics}
    \label{fig:basic_interaction_example}
\end{figure}

\section{Repetition operators in the semantic domain\label{sec:semantic_domain}}

Scheduling operators define compositions of traces obtained from enabling or forbidding the reordering of actions according to some scheduling policy. All three are associative (in addition, $\globalInterleaving$ is commutative) and admit $\{\epsilon\}$ as a neutral element. As a result, we can define (Kleene) closures of those operators to specify repetitions.

\begin{definition}[Kleene closures\label{def:kleene_closure}]For any $\diamond \in \{ \globalStrictSeq ,~\globalWeakSeq ,~ \globalInterleaving \}$ and any $T \in \mathcal{P}(\mathbb{T}_\Omega)$,

\begin{minipage}{0.45\textwidth}
the Kleene closure $T^{\diamond *}$ of $T$ is defined by:
\end{minipage}
\begin{minipage}{0.5\textwidth}
\[
\left\{
\begin{array}{lclr}
T^{\diamond 0}
&
=
&
\{ \epsilon \}
&
\\
T^{\diamond j}
&
=
&
T \diamond T^{\diamond (j-1)}
&
\textbf{for } j > 0
\\
T^{\diamond *}
&
=
&
\bigcup_{\substack{j \in \mathbb{N}}} T^{\diamond j}
&
\end{array}
\right.
\]
\end{minipage}
\end{definition}

The three Kleene closures $~^{\globalStrictSeq *}$, $~^{\globalWeakSeq *}$ and $~^{\globalInterleaving *}$
are respectively called, strict, weak and interleaving Kleene closures.

Within the K-closure $T^{\diamond *}$ we can find traces obtained from the repetition (using $\diamond$ as a scheduler) of any number of traces of $T$. 
In the example from Fig.\ref{fig:examples_weak_Kleene_closure} we consider a set $T$ containing two traces $\shortColOrange{l_1!m_1}.\shortColOrange{l_2!m_2}$ and $\shortColRed{l_2?m_1}$. The first 3 powersets of $T$ (i.e. $T^{\globalWeakSeq 0} \cup T^{\globalWeakSeq 1} \cup T^{\globalWeakSeq 2}$) are displayed, the rest of the weak K-closure of $T$ (i.e. $T^{\globalWeakSeq *}$) is represented by the $\cdots$.

\begin{figure}[ht]
    \centering
\[
\left\{
\begin{array}{l}
\shortColOrange{l_1!m_1}.\shortColOrange{l_2!m_2},\\
\shortColRed{l_2?m_1}
\end{array}
\right\}
^{\globalWeakSeq *}
= 
\left\{
\begin{array}{lclcl}
\begin{array}{l}
\epsilon,\\
\shortColOrange{l_1!m_1}.\shortColOrange{l_2!m_2},\\
\shortColRed{l_2?m_1},
\end{array}
&
~
&
\begin{array}{l}
\shortColOrange{l_1!m_1}.\shortColOrange{l_2!m_2}.\shortColOrange{l_1!m_1}.\shortColOrange{l_2!m_2},\\
\shortColOrange{l_1!m_1}.\shortColOrange{l_1!m_1}.\shortColOrange{l_2!m_2}.\shortColOrange{l_2!m_2},\\
\shortColOrange{l_1!m_1}.\shortColOrange{l_2!m_2}.\shortColRed{l_2?m_1},\\
\shortColRed{l_2?m_1}.\shortColRed{l_2?m_1},\\
\shortColRed{l_2?m_1}.\shortColOrange{l_1!m_1}.\shortColOrange{l_2!m_2},\\
\shortColOrange{l_1!m_1}.\shortColRed{l_2?m_1}.\shortColOrange{l_2!m_2},\\
\end{array}
&
~
&
\begin{array}{l}
\cdots
\end{array}
\end{array}
\right\}
\]
\caption{Example illustrating the weak Kleene closure}
\label{fig:examples_weak_Kleene_closure}
\end{figure}

Let us remark a useful property of the K-closures which is that $T \diamond T^{\diamond *} = T^{\diamond *}$.

Given a scheduling operator $\diamond \in \{\globalStrictSeq,~\globalWeakSeq,~\globalInterleaving\}$ whenever $a.t \in T_1 \diamond T_2$ (with $a$ and $t$ any action and trace and $T_1$ and $T_2$ any sets of traces), it may be so that action $a$ is taken from a trace $a.t'$ that belongs to either $T_1$ or $T_2$. In Def.\ref{def:restricted_scheduling_operators}, we define restricted versions of the scheduling operators so as to impose action $a$ to be taken from $T_1$.

\begin{definition}[Restricted scheduling operators\label{def:restricted_scheduling_operators}]
For any $\diamond \in \{\globalStrictSeq,~\globalWeakSeq,~\globalInterleaving\}$, we define the operator $\diamond^{\Lsh}$ such that for any sets of traces $T_1$ and $T_2$ we have:
\[
T_1 \diamond^{\Lsh} T_2
=
\left\{
~
t \in T_1 \diamond T_2
~\middle|~
(t = a.t')
\Rightarrow 
(
\exists~t_1 \in \mathbb{T}_\Omega, \text{ s.t. }
(a.t_1 \in T_1)
~\wedge~
(t \in \{t_1\} \diamond T_2)
)
~
\right\}
\]
\end{definition}

As an example, given $T_1 = \{\shortColRed{l_1!m}.\shortColOrange{l_1?m}\}$ and $T_2 = \{\shortColBlue{l_2!m}\}$, we have:
\[
\begin{array}{ccc}
T_1 \globalWeakSeq T_2 =
\left\{
\begin{array}{l}
\shortColRed{l_1!m}.\shortColOrange{l_1?m}.\shortColBlue{l_2!m},\\
\shortColRed{l_1!m}.\shortColBlue{l_2!m}.\shortColOrange{l_1?m},\\
\shortColBlue{l_2!m}.\shortColRed{l_1!m}.\shortColOrange{l_1?m}
\end{array}
\right\}
&
~~\text{ and }~~
&
T_1 \globalWeakSeq^\Lsh T_2 =
\left\{
\begin{array}{l}
\shortColRed{l_1!m}.\shortColOrange{l_1?m}.\shortColBlue{l_2!m},\\
\shortColRed{l_1!m}.\shortColBlue{l_2!m}.\shortColOrange{l_1?m}
\end{array}
\right\}
\end{array}
\]

We use the restriction on operators to define Head-First closures (abbr. HF-closure) of scheduling operators in Def.\ref{def:head_first_closures}.

\begin{definition}[Head-first closures\label{def:head_first_closures}]
For any $\diamond \in \{\globalStrictSeq,~\globalWeakSeq,~\globalInterleaving\}$, we define the Head-First closure of $\diamond$ as $^{\diamond^{\Lsh} *}$ i.e. the Kleene closure of the restricted $\diamond^{\Lsh}$ operator.
\end{definition}

In the following we will show that HF-closure and K-closure are equivalent for $\globalStrictSeq$ and $\globalInterleaving$ but that this is not the case for $\globalWeakSeq$.

\begin{lemma}\label{lem:strict_and_interleaving_kleene_closure_operational_charac_on_traces}
For any $\diamond \in \{\globalStrictSeq,~\globalInterleaving\}$, $T \in \mathcal{P}(\mathbb{T}_\Omega)$, $t$ in $\mathbb{T}_\Omega$ and $a \in \mathbb{A}_\Omega$ we have:
\[
(a.t \in T^{\diamond *}) \Rightarrow 
\left(
\exists~ t' \in \mathbb{T}_\Omega
\text{ s.t. }
(a.t' \in T)
\wedge~(t \in \{t'\} \diamond T^{\diamond *})
\right)
\]
\end{lemma}

\begin{proof}
By induction on $j$ given $a.t \in T^{\diamond j}$. In the case of $\globalInterleaving$ we use its commutativity.
\end{proof}

\begin{lemma}[Equivalence of HF \& K closures for $\globalStrictSeq$ \& $\globalInterleaving$]
\label{lem:equivalence_headfirst_and_kleene_closure_for_strict_and_interleaving}
For any set of traces $T$:
\[
\begin{array}{ccc}
T^{\globalStrictSeq^{\Lsh} *} = T^{\globalStrictSeq *}
&
~~\text{ and }~~
&
T^{\globalInterleaving^{\Lsh} *} = T^{\globalInterleaving *}
\end{array}
\]
\end{lemma}

\begin{proof}
By induction on a member trace $t$.
\end{proof}

Let us detail a counter example showing that the weak K-closure $~^{\globalWeakSeq *}$ and the weak HF-closure $~^{\globalWeakSeq^{\Lsh} *}$ are not equivalent.
Given $T=\{l_1!m_1.l_2?m_1,~l_2!m_2\}$, let us consider the powerset $T^{\globalWeakSeq 2}$ of $T$.
By definition, $\{\shortColRed{l_2!m_2}\} \globalWeakSeq \{\shortColViolet{l_1!m_1}.\shortColViolet{l_2?m_1}\} \subset T^{\globalWeakSeq 2}$. Here we can choose to take $\shortColViolet{l_1!m_1}$ as a first action and therefore $t = \shortColViolet{l_1!m_1}.\shortColRed{l_2!m_2}.\shortColViolet{l_2?m_1} \in T^{\globalWeakSeq 2}$. However, $t \not\in T \globalWeakSeq^\Lsh T = T^{\globalWeakSeq^\Lsh 2}$. Also, we have $t \not\in T^{\globalWeakSeq^\Lsh j}$ for any $j$ smaller or greater that $2$. This can be explained by not having the correct actions and/or not the right numbers of actions to reconstitute $t$. As a result $t \not\in T^{\globalWeakSeq^\Lsh *}$ and hence $T^{\globalWeakSeq *} \not\subseteq T^{\globalWeakSeq^{\Lsh} *}$. This example will be further illustrated in Sec.\ref{subsec:illustrative_example} with the help of the operational semantics.

\section{Syntax \& denotational semantics\label{sec:syntax_denotational_semantics}}

Interactions terms are defined inductively. Basic building blocks include the empty interaction $\varnothing$ which specifies the empty behavior $\epsilon$ (observation of no action) and any atomic action $a$, which specifies the single-element trace $a$ (observation of $a$). More complex behavior can then be specified inductively using:
\begin{itemize}
    \item binary constructors, which are $strict$, $seq$, $par$ and $alt$ (introduced in Sec.\ref{subsec:basic_interactions})
    \item unary constructors, which are $loop_X$ (the strict loop), $loop_H$ (the head loop up to the weak sequencing operator), $loop_S$ (the weak loop) and $loop_P$ (the parallel loop)
\end{itemize}

\begin{definition}[Interaction Language]\label{def:interaction_language}
We denote by $\mathbb{I}_\Omega$ the set of terms inductively defined by the set of operation symbols $\mathcal{F} = \mathcal{F}_0 \cup \mathcal{F}_1 \cup \mathcal{F}_2$ s.t.:
\begin{itemize}
    \item symbols or arity $0$ (constants) are $\mathcal{F}_0 = \{\varnothing\} \cup \mathbb{A}_\Omega$
    \item symbols of arity $1$ are $\mathcal{F}_1 = \{loop_X,~loop_H,~loop_S,~loop_P\}$
    \item symbols of arity $2$ are $\mathcal{F}_2 = \{strict,~seq,~par,~alt\}$
\end{itemize}
\end{definition}

In Def.\ref{def:interaction_language} we define our Interaction Language (IL) as a set of terms  $\mathbb{I}_\Omega$ inductively defined from $\mathcal{F}$ a finite set of symbols provided with arity in $\mathbb{N}$. The set of sets of traces $\mathcal{P}(\mathbb{T}_\Omega)$ admits the structure of a $\mathcal{F}$-algebra using operators introduced in Sec.\ref{sec:basic_interactions} and Sec.\ref{sec:semantic_domain}. The denotational semantics of interactions is then defined in Def.\ref{def:denotational_semantics} using the initial homomorphism associated to this $\mathcal{F}$-algebra.

\begin{definition}[Denotational semantics]\label{def:denotational_semantics}
$\mathcal{A} = (\mathcal{P}(\mathbb{T}_\Omega),~\{f^\mathcal{A} ~|~ f \in \mathcal{F}\})$ is the $\mathcal{F}$-algebra defined by its carrier $\mathcal{P}(\mathbb{T}_\Omega)$ and the following interpretations of the operation symbols in $\mathcal{F}$:
\[
\begin{array}{ccccc}
\begin{array}{lcl}
\varnothing^\mathcal{A} & = & \{\epsilon\} \\
a^\mathcal{A}           & = & \{a\} 
\end{array}
&
~~~~
&
\begin{array}{lcl}
strict^\mathcal{A} & = & \globalStrictSeq\\
seq^\mathcal{A}    & = & \globalWeakSeq\\
par^\mathcal{A}    & = & \globalInterleaving\\
alt^\mathcal{A}    & = & \cup
\end{array}
&
~~~~
&
\begin{array}{lcl}
loop_X^\mathcal{A} & = & ~^{\globalStrictSeq *}\\
loop_H^\mathcal{A} & = & ~^{\globalWeakSeq^\Lsh *}\\
loop_S^\mathcal{A} & = & ~^{\globalWeakSeq *}\\
loop_P^\mathcal{A} & = & ~^{\globalInterleaving *}
\end{array}
\end{array}
\]
 The denotational semantics $\sigma_d$ of $\mathbb{I}_\Omega$ is the unique $\mathcal{F}$-homomorphism $\sigma_d:\mathbb{I}_\Omega \rightarrow \mathcal{P}(\mathbb{T}_\Omega)$
 between the free term $\mathcal{F}$-algebra\footnote{The free term $\mathcal{F}$-algebra is defined by its carrier $\mathbb{I}_\Omega$ and by interpreting operation symbols of $\mathcal{F}$ as constructors of new terms: for $f \in \mathcal{F}$ of arity $j$, for $t_1, \ldots t_j \in \mathbb{I}_\Omega$, $f(t_1, \ldots t_j)$ is interpreted as itself.}
 $\mathcal{T}_\mathcal{F}$ and $\mathcal{A}$.
\end{definition}

The semantics of constants $\varnothing$ and $a \in \mathbb{A}_\Omega$ are sets containing a single element being respectively $\{\epsilon\}$ and $\{a\}$. The $strict$, $seq$, $par$ and $alt$ symbols are respectively associated to the $\globalStrictSeq$, $\globalWeakSeq$, $\globalInterleaving$ and union $\cup$ operators on sets of traces. Their use is shown on Fig.\ref{fig:basic_interaction_example}.
On the right of Fig.\ref{fig:basic_interaction_example}, we detail the semantics of the interaction drawn on the left. The resulting set of traces is given on the bottom right.
$loop_X$, $loop_H$, $loop_P$, $loop_S$ are semantically distinct (as shown on simple examples in \cite{revisiting_semantics_of_interactions_for_trace_validity_analysis}). From a system designer perspective, using either loop is motivated by different goals:
\begin{itemize}
    \item With $loop_X(i)$, each existing instance of a repeatable behavior (specified by $i$) must be executed entirely before any other can start. This implies that, at any given moment there can only exists zero or a single instance. $loop_X$ can therefore be used to specify some critical repeatable behavior of which there can only exist one instance at a time.
    \item With $loop_P(i)$, all existing instances can be executed concurrently w.r.t. one another, and, at any given moment, new instances can be created. $loop_P$ can therefore be used to specify protocols in which any number of new sessions can be created and run in parallel.
    \item With $loop_S(i)$, new instances can be created whenever the action triggering the instantiation occurs on a lifeline which is not occulted by previous instances. This can be roughly explained as follows: (1) on each individual lifeline only one instance can be active and (2) given that, for any such instance, a lifeline might "have finished" before the others then it is allowed to start another instance.
    $loop_S$ can therefore be used to specify repeatable behaviors that are sequential but that have no enforced synchronization mechanisms.
    \item The head loop $loop_H$ is associated to the weak HF-closure operator $~^{\globalWeakSeq^\Lsh *}$ which is an ad-hoc algebraic artifact and not a K-closure. We include it in our IL because it might correspond to a more intuitive understanding of sequential loops than $loop_S$, as illustrated in Sec.\ref{subsec:illustrative_example}.
\end{itemize}

\section{A structural operational semantics\label{sec:operational_semantics}}

In this section we present a structural operational semantics for interactions in the style of process calculus \cite{process_algebra_with_explicit_termination}. It relies on the definition (by structural induction) of two predicates: "$i \downarrow$" (the termination predicate) indicates that interaction term $i$ accepts the empty trace and "$i \xrightarrow{a} i'$" (the execution relation) indicates that traces $a.t$ such that $t$ is accepted by $i'$ are accepted by $i$.
We define $\downarrow$ in Sec.\ref{subsec:termination}. The relation $\rightarrow$ allows the determination, for any interaction $i$, of which actions $a$ can be immediately executed, and, of potential follow-up interactions $i'$ which express continuations $t$ of traces $a.t$ accepted by $i$. Defining an execution relation $\rightarrow$ is a staple of process calculus \cite{process_algebra_with_explicit_termination}.
However, as we want to define such a similar relation $\rightarrow$  for our IL, particular care is needed when dealing with weak sequencing. Hence we introduce intermediate notions in Sec.\ref{subsec:evades_prune} before defining $\rightarrow$ in Sec.\ref{subsec:execution_relation}.

\subsection{Termination\label{subsec:termination}}

By reasoning on the structure of an interaction term
$i$, we can determine whether or not the empty trace $\epsilon$ belongs to its semantics. When this holds, we say that $i$ terminates and use the notation $i \downarrow$ as in \cite{process_algebra_with_explicit_termination}.
$\downarrow$ is defined as an inductive predicate in Def.\ref{def:termination} and its characterization w.r.t. the semantics $\sigma_d$ of interactions is given in Lem.\ref{lem:sem_de_terminates}.

\begin{definition}[Termination]\label{def:termination}
The predicate $\downarrow \subset \mathbb{I}_\Omega$ is s.t. for any $i_1$ and $i_2$ from $\mathbb{I}_\Omega$, any $f \in \{strict,seq,par\}$ and any $k \in \{X,H,S,P\}$ we have:

{
\centering
\begin{minipage}{1.75cm}
\begin{prooftree}
\AxiomC{\phantom{$\top$}}
\UnaryInfC{$\varnothing \downarrow$}
\end{prooftree}
\end{minipage}
\begin{minipage}{2.75cm}
\begin{prooftree}
\AxiomC{$i_1 \downarrow$}
\UnaryInfC{$alt(i_1,i_2) \downarrow$}
\end{prooftree}
\end{minipage}
\begin{minipage}{2.75cm}
\begin{prooftree}
\AxiomC{$i_2 \downarrow$}
\UnaryInfC{$alt(i_1,i_2) \downarrow$}
\end{prooftree}
\end{minipage}
\begin{minipage}{2.75cm}
\begin{prooftree}
\AxiomC{$i_1 \downarrow$}
\AxiomC{$i_2 \downarrow$}
\BinaryInfC{$f(i_1,i_2) \downarrow$}
\end{prooftree}
\end{minipage}
\begin{minipage}{2.75cm}
\begin{prooftree}
\AxiomC{\phantom{$\top$}}
\UnaryInfC{$loop_k(i_1) \downarrow$}
\end{prooftree}
\end{minipage}\\
}
\end{definition}

All rules of Def.\ref{def:termination} are evident. The empty interaction $\varnothing$ only accepts $\epsilon$, and thus terminates. An interaction with a loop at its root terminates because it is possible to repeat zero times its content.
As $alt(i_1,i_2)$ specifies a choice, it terminates iff either $i_1$ or $i_2$ terminates.
An interaction of the form $f(i_1,i_2)$, with $f$ being a scheduling constructor, terminates iff both $i_1$ and $i_2$ terminate.

\begin{lemma}[Termination w.r.t. $\sigma_d$]\label{lem:sem_de_terminates}
For any $i \in \mathbb{I}_\Omega$ we have $(i \downarrow ) \Leftrightarrow (\epsilon \in \sigma_d(i))$
\end{lemma}

\begin{proof}
By induction on the term structure of interactions.
\end{proof}

\subsection{Dealing with weak-sequencing using evasion \& pruning\label{subsec:evades_prune}}

Weak sequencing only allows interleavings between actions that occur on different lifelines. As a result, within an interaction of the form $i = seq(i_1,i_2)$, some actions that can be executed in $i_2$ may also be executed in $seq(i_1,i_2)$. In other words, given a trace $a.t \in \sigma_d(i)$, action $a$ might correspond to an action expressed by $i_2$. This is however conditioned by the ability of $i_1$ to express traces that have no conflict w.r.t. $a$ so that $a$ may be placed in front of what is expressed by $i_1$ when recomposing $a.t$.

In the previous section we have seen that the termination predicate $i\downarrow$ states that interaction $i$ is able to express the empty trace. We define "evasion" as a similar, although weaker, notion than "termination" that can be described as a form of local termination.
For a lifeline $l$, we say that $i$ evades $l$, denoted by $i \evadesLf l$, if $i$ accepts at least one trace that does not contain actions occurring on $l$. Evasion is defined inductively in Def.\ref{def:evasion}. For any $i \in \mathbb{I}_\Omega$ and $l \in L$, it is always decidable whether $i \evadesLf l$ or $\neg(i \evadesLf l)$. We may then denote by $i \collidesLf l$ the fact that $\neg(i \evadesLf l)$ and say that $i$ collides with $l$.

\begin{definition}[Evasion]\label{def:evasion}
The predicate $\evadesLf \subset \mathbb{I}_\Omega \times L$ is s.t. for any $i_1$ and $i_2$ from $\mathbb{I}_\Omega$, any $l\in L$, any $a \in \mathbb{A}_\Omega$, any $f \in \{strict,seq,par\}$ and any $k \in \{X,H,S,P\}$ we have:

{
\centering
\begin{minipage}{2cm}
\begin{prooftree}
\AxiomC{\phantom{$\theta(a) \neq l$}}
\UnaryInfC{$\varnothing \evadesLf l$}
\end{prooftree}
\end{minipage}
\begin{minipage}{2cm}
\begin{prooftree}
\AxiomC{$\theta(a) \neq l$}
\UnaryInfC{$a \evadesLf l$}
\end{prooftree}
\end{minipage}

\vspace*{.1cm}

\begin{minipage}{2.75cm}
\begin{prooftree}
\AxiomC{$i_1 \evadesLf l$}
\UnaryInfC{$alt(i_1,i_2) \evadesLf l$}
\end{prooftree}
\end{minipage}
\begin{minipage}{2.75cm}
\begin{prooftree}
\AxiomC{$i_2 \evadesLf l$}
\UnaryInfC{$alt(i_1,i_2) \evadesLf l$}
\end{prooftree}
\end{minipage}
\begin{minipage}{3.5cm}
\begin{prooftree}
\AxiomC{$i_1 \evadesLf l$}
\AxiomC{$i_2 \evadesLf l$}
\BinaryInfC{$f(i_1,i_2) \evadesLf l$}
\end{prooftree}
\end{minipage}
\begin{minipage}{2.75cm}
\begin{prooftree}
\AxiomC{\phantom{$\evadesLf$}}
\UnaryInfC{$loop_k(i_1) \evadesLf l$}
\end{prooftree}
\end{minipage}\\
}
\end{definition}

The empty interaction $\varnothing$ evades any lifeline as $\epsilon$ contains no action. 
An interaction reduced to a single action $a$ evades $l$ iff $a$ does not occur on $l$. As for termination, an interaction with a loop at its root evades any lifeline because it accepts $\epsilon$.
Choice and scheduling operators are also handled in the same manner as for the termination predicate.

\begin{lemma}[Evasion w.r.t. $\sigma_d$]\label{lem:sem_de_evades_no_conflict}
For any $l \in L$ and $i \in \mathbb{I}_\Omega$, $(i \evadesLf l)
\Leftrightarrow 
(\exists~t \in \sigma_d(i), \neg(t \doubleVerticalTimes l))$
\end{lemma}

\begin{proof}
By induction on the term structure of interactions.
\end{proof}

Let us remark that, for any $i \in \mathbb{I}_\Omega$, if $i \downarrow$ then $\forall~ l \in L$, $i \evadesLf l$. Indeed, $\epsilon$ has no conflict w.r.t. any $l$. However the opposite does not hold:
it suffices to consider $L=\{ l_1,l_2\}$ and $i=alt(l_1!m,l_2!m)$ and observe that $\forall~ l \in L$, $i \evadesLf l$ holds while
$i \downarrow$ does not.

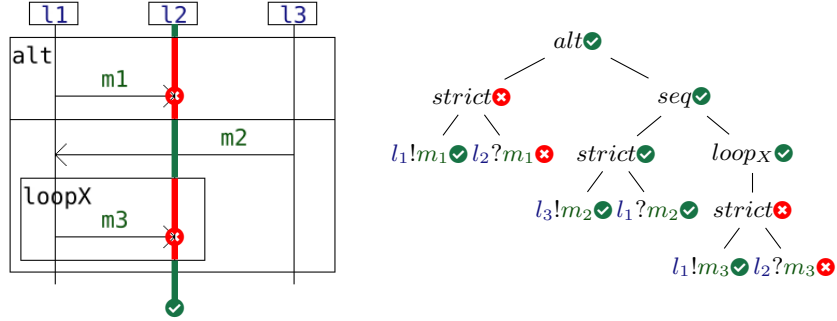
\begin{figure}[ht]
    \centering
\begingroup
\setlength{\tabcolsep}{1pt} 
\renewcommand{\arraystretch}{.1} 
\begin{tabular}{cc}
\makecell{
\scalebox{.8}{
\begin{tikzpicture}
\node (int) at (0,0) {\includegraphics[scale=.5]{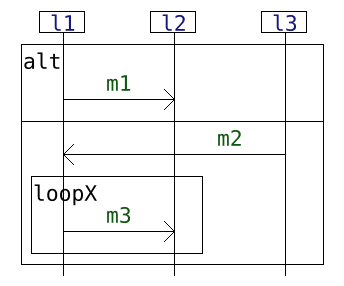}};
\draw[-,line width=3pt,darkspringgreen] (1.5pt,1.96) -- (1.5pt,1.75);
\draw[-,line width=3pt,red]             (1.5pt,1.73) -- (1.5pt,.9);
\node (mark1) at (1.5pt,.775) {\small\textcolor{red}{\faTimesCircle}};
\draw[-,line width=3pt,red]             (1.5pt,.7) -- (1.5pt,.39);
\draw[-,line width=3pt,darkspringgreen] (1.5pt,.37) -- (1.5pt,-.58);
\draw[-,line width=3pt,red] (1.5pt,-.6) -- (1.5pt,-1.44);
\node (mark2) at (1.5pt,-1.56) {\small\textcolor{red}{\faTimesCircle}};
\draw[-,line width=3pt,red] (1.5pt,-1.65) -- (1.5pt,-1.94);
\draw[-,line width=3pt,darkspringgreen] (1.5pt,-1.96) -- (1.5pt,-2.7);
\node (mark3) at (1.5pt,-2.75) {\small\textcolor{darkspringgreen}{\faCheckCircle}};
\end{tikzpicture}
}
}
&
\makecell{
\scalebox{1}{
\begin{tikzpicture}
\node (o) { $alt${\scriptsize\textcolor{darkspringgreen}{\faCheckCircle}} } 
                [sibling distance=2.8cm,level distance=0.75cm]
  child { node (o1) {$strict${\scriptsize\textcolor{red}{\faTimesCircle}}} 
                [sibling distance=1.1cm]
    child { node (o11) {$\hlf{l_1}!\hms{m_1}${\scriptsize\textcolor{darkspringgreen}{\faCheckCircle}}} }
    child { node (o12) {$\hlf{l_2}?\hms{m_1}${\scriptsize\textcolor{red}{\faTimesCircle}}} }
  }
  child { node (o2) {$seq${\scriptsize\textcolor{darkspringgreen}{\faCheckCircle}}} [sibling distance=1.8cm]
    child { node (o21) {$strict${\scriptsize\textcolor{darkspringgreen}{\faCheckCircle}}} 
                [sibling distance=1.1cm]
      child { node (o211) {$\hlf{l_3}!\hms{m_2}${\scriptsize\textcolor{darkspringgreen}{\faCheckCircle}}} }
      child { node (o212) {$\hlf{l_1}?\hms{m_2}${\scriptsize\textcolor{darkspringgreen}{\faCheckCircle}}} }
    }
    child { node (o22) {$loop_X${\scriptsize\textcolor{darkspringgreen}{\faCheckCircle}}} 
      child { node (o221) {$strict${\scriptsize\textcolor{red}{\faTimesCircle}}} 
                [sibling distance=1.1cm]
        child { node (o2211) {$\hlf{l_1}!\hms{m_3}${\scriptsize\textcolor{darkspringgreen}{\faCheckCircle}}} }
        child { node (o2212) {$\hlf{l_2}?\hms{m_3}${\scriptsize\textcolor{red}{\faTimesCircle}}} }
      }
    }
  };
\end{tikzpicture}
}
}
\end{tabular}
\endgroup 
    \caption{Illustration of the evasion predicate (here w.r.t. lifeline $l_2$)}
    \label{fig:illustration_evasion_predicate}
\end{figure}

The application of the evasion predicate (w.r.t. lifeline $l_2$) is illustrated on Fig.\ref{fig:illustration_evasion_predicate}. On the right is represented the syntaxic structure of an interaction $i$, and, on the left, the corresponding drawing as a sequence diagram. On the syntax tree, the nodes are decorated with symbols {\scriptsize\textcolor{darkspringgreen}{\faCheckCircle}} to signify that the sub-interaction underneath that node evades $l_2$ or {\scriptsize\textcolor{red}{\faTimesCircle}} to indicate that this is not the case (i.e. that it collides with $l_2$). Starting from the leaves we can decorate all nodes and conclude once the root is reached. By taking the right branch of the alternative and by choosing not to instantiate the loop, we can see that $i$ accepts some traces that have no conflict w.r.t. lifeline $l_2$ (in our case, only the trace $l_3!m_2.l_1?m_2$). As a result the interaction $i$ verifies $i \evadesLf l_2$.
On the diagram representation, evasion can be illustrated by drawing a line over the lifeline of interest. This line can be decomposed into several segments that correspond to specific areas of the diagram (operands) and that are colored either in green or in red. The coloration of the segment depends on whether the sub-interaction corresponding to the operand evades or collides with the lifeline of interest.

With evasion we can determine whether or not an action $a$ that is executable in $i_2$ i.e. s.t. $i_2 \xrightarrow{a} i_2'$ is also executable in $i = seq(i_1,i_2)$. However, this is not enough to define a rule $seq(i_1,i_2) \xrightarrow{a} i'$ which is compatible with the semantics $\sigma_d$. $i'$ must specify continuations $t$ s.t. $a.t \in \sigma_d(i)$. By definition of $\sigma_d$, continuation traces $t$ have to be build from traces $t_1 \in \sigma_d(i_1)$ and $t_2$ such that $\neg(t_1 \doubleVerticalTimes \theta(a))$ and $a.t_2 \in \sigma_d(i_2)$. By defining $i_1'$ as the interaction which expresses exactly those traces $t_1$ s.t. $\neg(t_1 \doubleVerticalTimes \theta(a))$ we may produce a rule $seq(i_1,i_2) \xrightarrow{a} seq(i_1',i_2')$ that is compatible with $\sigma_d$. The computation of $i_1'$ is ensured by a process that we call pruning.

Pruning is defined as an inductive relation $\isPruneBase$ s.t. $i \isPruneOf{l} i'$ indicates that the pruning of $i \in \mathbb{I}_\Omega$ w.r.t. $l \in L$ yields $i' \in \mathbb{I}_\Omega$.
Pruning is defined so has to preserve "a maxima" the original semantics of $i$ so that $\sigma_d(i') \subseteq \sigma_d(i)$ is the maximum subset of $\sigma_d(i)$ that contains no trace conflicting with $l$ (see Lem.\ref{lem:sem_de_pruned}).

\begin{definition}[Pruning]\label{def:pruning_relation}
The pruning relation $\isPruneBase \subset \mathbb{I}_\Omega \times L \times \mathbb{I}_\Omega$ is s.t. for any $l \in L$, any $f \in \{strict,seq,par\}$ and any $k \in \{X,H,S,P\}$:

{
\centering
\begin{minipage}{2.5cm}
\begin{prooftree}
\AxiomC{\vphantom{$\isPruneOf{l}$}}
\UnaryInfC{$\varnothing \isPruneOf{l} \varnothing$}
\end{prooftree}
\end{minipage}
\begin{minipage}{2.5cm}
\begin{prooftree}
\AxiomC{$\theta(a) \neq l$\vphantom{$\isPruneOf{l}$}}
\UnaryInfC{$a \isPruneOf{l} a$}
\end{prooftree}
\end{minipage}
\begin{minipage}{4cm}
\begin{prooftree}
\AxiomC{$i_1 \isPruneOf{l} i_1'$}
\AxiomC{$i_2 \isPruneOf{l} i_2'$}
\BinaryInfC{$f(i_1,i_2) \isPruneOf{l} f(i_1',i_2')$}
\end{prooftree}
\end{minipage}

\vspace*{.1cm}

\begin{minipage}{4cm}
\begin{prooftree}
\AxiomC{$i_1 \isPruneOf{l} i_1'$}
\AxiomC{$i_2 \isPruneOf{l} i_2'$}
\BinaryInfC{$alt(i_1,i_2) \isPruneOf{l} alt(i_1',i_2')$}
\end{prooftree}
\end{minipage}
\begin{minipage}{4cm}
\begin{prooftree}
\AxiomC{$i_1 \isPruneOf{l} i_1'$}
\RightLabel{$i_2 \collidesLf l$}
\UnaryInfC{$alt(i_1,i_2) \isPruneOf{l} i_1'$}
\end{prooftree}
\end{minipage}
\begin{minipage}{4cm}
\begin{prooftree}
\AxiomC{$i_2 \isPruneOf{l} i_2'$}
\RightLabel{$i_1 \collidesLf l$}
\UnaryInfC{$alt(i_1,i_2) \isPruneOf{l} i_2'$}
\end{prooftree}
\end{minipage}

\vspace*{.1cm}

\begin{minipage}{5cm}
\begin{prooftree}
\AxiomC{$i_1 \isPruneOf{l} i_1'$}
\UnaryInfC{$loop_k(i_1) \isPruneOf{l} loop_k(i_1')$}
\end{prooftree}
\end{minipage}
\begin{minipage}{5cm}
\begin{prooftree}
\AxiomC{$\phantom{\isPruneOf{l}}$}
\RightLabel{$i_1 \collidesLf l$}
\UnaryInfC{$loop_k(i_1) \isPruneOf{l} \varnothing$}
\end{prooftree}
\end{minipage}
\\
}
\end{definition}

Let us immediately remark that evasion and pruning are intertwined notions. Indeed, as per Lem.\ref{lem:pruning_existence_unicity} evasion is equivalent to the existence and unicity of a pruned interaction.
We could define pruning without evasion. However, so as to draw parallels w.r.t. the termination predicate and in order to separate concerns and ease the understanding of proofs we have defined both notions separately.

\begin{lemma}[Conditional existence \& unicity for pruning]\label{lem:pruning_existence_unicity}
For any $i \in \mathbb{I}_\Omega$ and any $l \in L$:
\[
( i \evadesLf l ) \Leftrightarrow (\exists!~i' \in \mathbb{I}_\Omega \text{ s.t. } i \isPruneOf{l} i')
\]
\end{lemma}

\begin{proof}
Immediate, by induction on the term structure of interactions.
\end{proof}

Let us comment on the rules defining the pruning relation.
We have $\varnothing \isPruneOf{l} \varnothing$ because the semantics of $\varnothing$ being $\{\epsilon\}$, there are no conflicts w.r.t. $l$. 
Any action $a \in \mathbb{A}_\Omega$ is prunable iff $\theta(a) \neq l$. In such a case, $a$ needs not be eliminated and thus $a \isPruneOf{l} a$. For $i=alt(i_1,i_2)$ to be prunable we must have either or both of $i_1 \evadesLf l$ or $i_2 \evadesLf l$. If both branches evade $l$ they can be pruned and kept as alternatives in the new interaction term. If only a single one does, we only keep the pruned version of this single branch. For any scheduling constructor $f$, if $i=f(i_1,i_2)$, in order to have $i \evadesLf l$ we must have both $i_1 \evadesLf l$ and $i_2 \evadesLf l$. In that case the unique interaction $i'$ such that $i \isPruneOf{l} i'$ is simply defined as the scheduling, using $f$, of the pruned versions of $i_1$ and $i_2$. Finally, for loops, i.e. $i=loop_k(i_1)$ with $k \in \{X,H,S,P\}$, we distinguish two cases: (a) if $i_1 \not\evadesLf l$ then any execution of $i_1$ will yield a trace conflicting $l$. Therefore it is necessary to forbid repetition;
(b) if $i_1 \evadesLf l$ then it is not necessary to forbid repetition, given that $i_1$ can be pruned as $i_1 \isPruneOf{l} i_1'$ with $i_1'$ not expressing traces conflicting $l$. This being the modification which preserves a maximum amount of traces, we have $loop_k(i_1) \isPruneOf{l} loop_k(i_1')$.

\begin{figure}
    \centering
\begin{tabular}{ccc}
\makecell{
\includegraphics[scale=.3]{figures/4/ex_evasion.png}
}
&
\makecell{
\scalebox{1}{
\begin{tikzpicture}[every node/.style = {shape=rectangle, align=center}]
\node (o) { $alt$ } 
                [sibling distance=2.2cm,level distance=0.75cm]
  child { node (o1) {$strict$} 
                [sibling distance=.9cm]
    child { node (o11) {$\hlf{l_1}!\hms{m_1}$} }
    child { node (o12) {$\hlf{l_2}?\hms{m_1}$} }
  }
  child { node (o2) {$seq$} [sibling distance=1.4cm]
    child { node (o21) {$strict$} 
                [sibling distance=.9cm]
      child { node (o211) {$\hlf{l_3}!\hms{m_2}$} }
      child { node (o212) {$\hlf{l_1}?\hms{m_2}$} }
    }
    child { node (o22) {$loop_X$} 
      child { node (o221) {$strict$} 
                [sibling distance=.9cm]
        child { node (o2211) {$\hlf{l_1}!\hms{m_3}$} }
        child { node (o2212) {$\hlf{l_2}?\hms{m_3}$} }
      }
    }
  };
\draw[blue,very thick] (o12.south east) -- (o12.north west);
\draw[blue,very thick] (o12.south west) -- (o12.north east);
\draw[blue,very thick] (o1.south east) -- (o1.north west);
\draw[blue,very thick] (o1.south west) -- (o1.north east);
\draw[blue,very thick] (o2212.south east) -- (o2212.north west);
\draw[blue,very thick] (o2212.south west) -- (o2212.north east);
\draw[blue,very thick] (o221.south east) -- (o221.north west);
\draw[blue,very thick] (o221.south west) -- (o221.north east);
\draw[blue,very thick] (o22.east) -- (o22.west);
\node[blue,thick] (varno) at ([xshift=.75cm] o22.east) {$\varnothing$};
\draw[->,blue,very thick] (varno) edge ([xshift=2.5pt] o22.east);
\draw[blue,very thick] (o.east) -- (o.west);
\draw[->,blue,very thick] ([xshift=-5pt,yshift=0pt] o2.north east) to [bend right=45] ([xshift=2.5pt] o.east);
\end{tikzpicture}}
}
&
\makecell{
\includegraphics[scale=.3]{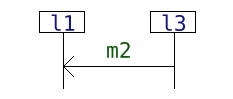}
}\\
\makecell{
\scriptsize\textit{before pruning}
}
&
\makecell{
\scriptsize\textit{pruning w.r.t $\hlf{l_3}$}
}
&
\makecell{
\scriptsize\textit{after pruning}
}
\end{tabular}
\caption{Illustration of pruning}
\label{fig:pruning_example}
\end{figure}

We have seen that the interaction $i$ of Fig.\ref{fig:illustration_evasion_predicate} satisfies $i \evadesLf l_2$. Therefore Lem.\ref{lem:pruning_existence_unicity} implies the existence of a unique $i'$ s.t. $i \isPruneOf{l_2} i'$. Fig.\ref{fig:pruning_example} illustrates the computation of $i'$.
The blue lines represent the modifications in the syntax of $i$ that occur during its pruning into $i'$.
On Fig.\ref{fig:illustration_evasion_predicate} we decorated sub-interactions of $i$ with {\scriptsize\textcolor{red}{\faTimesCircle}} whenever they did not evade $l_2$. During pruning, those sub-interactions must be eliminated given that the resulting term must not express actions occurring on $l_2$.
Hence, on Fig.\ref{fig:pruning_example}, we have crossed in blue the problematic sub-interactions.
The root node is an $alt$. Let us note $i = alt(i_1,i_2)$. On Fig.\ref{fig:illustration_evasion_predicate} we have seen that we have $i_1 \collidesLf l_2$ and $i_2 \evadesLf l_2$. Therefore we have $i \isPruneOf{l_2} i_2'$ with $i_2'$ being such that $i_2 \isPruneOf{l_2} i_2'$. This selection of the right branch of the $alt$ is illustrated on Fig.\ref{fig:pruning_example} by the curved arrow which "replaces" the $alt$ by the $seq$ on its bottom right.
There remains to determine $i_2'$ s.t. $i_2 \isPruneOf{l_2} i_2'$. At the root of $i_2$ we have a $seq$. Let us note $i_2 = seq(i_A,i_B)$. As per Fig.\ref{fig:illustration_evasion_predicate} we have both $i_A \evadesLf l_2$ and $i_B \evadesLf l_2$ and therefore $i_2' = seq(i_A',i_B')$ such that $i_A \isPruneOf{l_2} i_A'$ and $i_B \isPruneOf{l_2} i_B'$. Underneath $i_A$, no actions occur on $l_2$ and hence $i_A' = i_A$. At the root of $i_B$ we have a $loop_X$. Let us note $i_B = loop_X(i_C)$. As per Fig.\ref{fig:illustration_evasion_predicate} we have $i_C \collidesLf l_2$ and therefore $i_B' = \varnothing$ which is illustrated on Fig.\ref{fig:pruning_example} by the $\shortColBlue{\leftarrow \varnothing}$ in blue, which "replaces" the $loop_X$ by $\varnothing$.
Finally there remain $i' = seq(i_A,\varnothing)$, which is drawn as a SD on the right of Fig.\ref{fig:pruning_example}.

Lem.\ref{lem:sem_de_pruned} states that given 
$i \isPruneOf{l} i'$, the pruned interaction $i'$
exactly specifies all the executions of $i$ that do not involve $l$.

\begin{lemma}[Pruning w.r.t. $\sigma_d$]\label{lem:sem_de_pruned}
For any $l \in L$ and any $i$ and $i'$ from $\mathbb{I}_\Omega$:
\[
(i \isPruneOf{l} i')
\Rightarrow
(\sigma_d(i') = \{ t \in \sigma_d(i)~|~ \neg (t \doubleVerticalTimes l) \})
\]
\end{lemma}

\begin{proof}
By induction on the term structure of interactions.
\end{proof}

In the next section we make use of pruning so as to formally define an execution relation $\rightarrow$ for our IL and we use it to formulate a structural operational semantics in the style of process algebra.

\subsection{Execution relation \& operational semantics\label{subsec:execution_relation}}

A structural operational semantics, presented in the style of Plotkin \cite{a_structural_approach_to_operational_semantics} allows determining accepted traces $t \in \sigma_d(i_1)$ such that $t = a_1.\cdots.a_n$ through the assertion of a succession of predicates of the form $i_j \xrightarrow{a_j} i_{j+1}$ representing the evolution of the system. By expressing action $a_j$, the system goes from being modelled by $i_j$ to being modelled by $i_{j+1}$. We present this execution relation "$\rightarrow$" as an inductive predicate, in a similar fashion as the pruning relation "$\isPruneBase$".

\begin{definition}[Execution relation]\label{def:execution_relation}
The execution relation $\rightarrow \subset \mathbb{I}_\Omega \times \mathbb{A}_\Omega  \times \mathbb{I}_\Omega$ is s.t.:

{
\centering
\begin{minipage}{2cm}
\begin{prooftree}
\AxiomC{\phantom{$\xrightarrow{a}$}}
\UnaryInfC{$a \xrightarrow{a} \varnothing$}
\end{prooftree}
\end{minipage}
\begin{minipage}{3cm}
\begin{prooftree}
\AxiomC{$i_1 \xrightarrow{a} i'_1$}
\UnaryInfC{$alt(i_1,i_2) \xrightarrow{a} i'_1$}
\end{prooftree}
\end{minipage}
\begin{minipage}{3cm}
\begin{prooftree}
\AxiomC{$i_2 \xrightarrow{a} i'_2$}
\UnaryInfC{$alt(i_1,i_2) \xrightarrow{a} i'_2$}
\end{prooftree}
\end{minipage}

\vspace{0.1cm}

\begin{minipage}{5cm}
\begin{prooftree}
\AxiomC{$i_1 \xrightarrow{a} i'_1$}
\UnaryInfC{$par(i_1,i_2) \xrightarrow{a} par(i'_1,i_2)$}
\end{prooftree}
\end{minipage}
\begin{minipage}{5cm}
\begin{prooftree}
\AxiomC{$i_2 \xrightarrow{a} i'_2$}
\UnaryInfC{$par(i_1,i_2) \xrightarrow{a} par(i_1,i'_2)$}
\end{prooftree}
\end{minipage}

\vspace{0.1cm}

\begin{minipage}{5cm}
\begin{prooftree}
\AxiomC{$i_1 \xrightarrow{a} i'_1$}
\UnaryInfC{$strict(i_1,i_2) \xrightarrow{a} strict(i'_1,i_2)$}
\end{prooftree}
\end{minipage}
\begin{minipage}{5cm}
\begin{prooftree}
\AxiomC{$i_2 \xrightarrow{a} i'_2$}
\RightLabel{$i_1 \downarrow$}
\UnaryInfC{$strict(i_1,i_2) \xrightarrow{a} i'_2$}
\end{prooftree}
\end{minipage}

\vspace{0.1cm}

\begin{minipage}{5cm}
\begin{prooftree}
\AxiomC{$i_1 \xrightarrow{a} i'_1$}
\UnaryInfC{$seq(i_1,i_2) \xrightarrow{a} seq(i'_1,i_2)$}
\end{prooftree}
\end{minipage}
\begin{minipage}{5cm}
\begin{prooftree}
\AxiomC{$i_1 \isPruneOf{\theta(a)} i_1'$}
\AxiomC{$i_2 \xrightarrow{a} i_2'$}
\BinaryInfC{$seq(i_1,i_2) \xrightarrow{a} seq(i_1',i_2')$}
\end{prooftree}
\end{minipage}

\vspace{0.1cm}

\begin{minipage}{6cm}
\begin{prooftree}
\AxiomC{$i_1 \xrightarrow{a} i_1'$}
\UnaryInfC{$loop_X(i_1) \xrightarrow{a} strict(i_1',loop_X(i_1))$}
\end{prooftree}
\end{minipage}
\begin{minipage}{6cm}
\begin{prooftree}
\AxiomC{$i_1 \xrightarrow{a} i_1'$}
\UnaryInfC{$loop_H(i_1) \xrightarrow{a} seq(i_1',loop_H(i_1))$}
\end{prooftree}
\end{minipage}
 
\vspace{0.1cm}
 
\begin{minipage}{6cm}
\begin{prooftree}
\AxiomC{$i_1 \xrightarrow{a} i_1'$}
\AxiomC{$loop_S(i_1) \isPruneOf{\theta(a)} i'$}
\BinaryInfC{$loop_S(i_1) \xrightarrow{a} seq(i',seq(i_1',loop_S(i_1)))$}
\end{prooftree}
\end{minipage}
\begin{minipage}{6cm}
\begin{prooftree}
\AxiomC{$i_1 \xrightarrow{a} i_1'$}
\UnaryInfC{$loop_P(i_1) \xrightarrow{a} par(i_1',loop_P(i_1))$}
\end{prooftree}
\end{minipage}\\
}
\end{definition}

The formulation of most of the rules is very similar to what can be found in process algebras. The notable differences mainly concern the rules handling weak sequencing and loops.
In an interaction reduced to an action $a$, $a$ may be executed with $a \xrightarrow{a} \varnothing$ so that what remains to be executed is the empty interaction $\varnothing$. If within $i=alt(i_1,i_2)$, action $a$ can be executed in either $i_1$ or $i_2$ with either $i_1 \xrightarrow{a} i_1'$ or $i_2 \xrightarrow{a} i_2'$ then it may be executed in $i$ and the resulting interaction is either $i_1'$ or $i_2'$.
For $i=par(i_1,i_2)$, if we have either $i_1 \xrightarrow{a} i_1'$ or $i_2 \xrightarrow{a} i_2'$ then $a$ may be executed in $i$ and the resulting interaction naturally is either $par(i_1',i_2)$ or $par(i_1,i_2')$. 
Executing actions on the left of either a $strict$ or a $seq$ follow the same rule as in the case of a $par$ because no precedence relations are enforced on the left-hand side.
However, it is only possible to execute an action $a$ on the right of $i=strict(i_1,i_2)$ if $i_1$ terminates. Indeed, in that case $i_1$ may express the empty trace $\epsilon$ as per Lem.\ref{lem:sem_de_terminates} and nothing prevents $a$ to be the first action to be executed. The resulting interaction is then $i_2'$ given that we force $i_1$ to express $\epsilon$.
Likewise, within $i=seq(i_1,i_2)$ there is a condition for executing an action $a$ on the right. This condition is that $i_1 \evadesLf \theta(a)$, which, as per Lem.\ref{lem:pruning_existence_unicity} is implied by the condition $i_1 \isPruneOf{\theta(a)} i_1'$. Given $i_2'$ s.t. $i_2 \xrightarrow{a} i_2'$ we therefore have $i \xrightarrow{a} seq(i_1',i_2')$ given that we prune the left sub-interaction and execute the action on the right sub-interaction.

The rules for $loop_X$, $loop_H$ and $loop_P$ reflect the definitions of the corresponding HF-closures which respectively are $~^{\globalStrictSeq^\Lsh *}$, $~^{\globalWeakSeq^\Lsh *}$ and $~^{\globalInterleaving^\Lsh *}$. Let us indeed consider $(k,\diamond,f) \in \{(X,\globalStrictSeq,strict),(H,\globalWeakSeq,seq),(P,\globalInterleaving,par)\}$. Whenever $i \xrightarrow{a} i'$ we have $loop_k(i) \xrightarrow{a} f(i',loop_k(i))$. As a result, any $t \in \sigma_d(f(i',loop_k(i)))$ is s.t. $t \in \{t_1\} \diamond \sigma_d(loop_k(i))$ for a certain $t_1 \in \sigma_d(i')$. Given that $a.t_1 \in \sigma_d(i)$, this corresponds to choosing action $a$ from the first iteration of the loop i.e. $a.t \in \{a.t_1\} \diamond^\Lsh \sigma_d(loop_k(i)) \subset \sigma_d(i) \diamond^\Lsh \sigma_d(loop_k(i))$,
which coincides with using the restricted operator $\diamond^\Lsh$ as a scheduler.
$loop_H$ is explicitly associated to $~^{\globalWeakSeq^\Lsh *}$ and thus the formulation of its rule is self-evident. In the case of $loop_X$ and $loop_P$ it is the fact that the HF and K-closures of $\globalStrictSeq$ and $\globalInterleaving$ are equivalent (as per Lem.\ref{lem:equivalence_headfirst_and_kleene_closure_for_strict_and_interleaving}) which enables their respective rules to be formulated in this manner.

The rule for $loop_S$ allows for the first action $a$ of a trace $a.t$ to be taken from a later iteration of the loop.
Let us consider $i = loop_S(i_1)$ and $a.t \in \sigma_d(i)$. The rule is formulated such that $t \in \sigma_d(seq(i',seq(i_1',i)))$ with $i \isPruneOf{\theta(a)} i'$ and $i_1 \xrightarrow{a} i_1'$. Indeed, action $a$ is expressed by sub-interaction $i_1$ so there exists $i_1'$ s.t. $i_1 \xrightarrow{a} i_1'$. Also, given that $i$ is a loop, it is always prunable (Lem.\ref{lem:pruning_existence_unicity}) so there exists $i'$ s.t. $i \isPruneOf{\theta(a)} i'$. The fact that $t \in \sigma_d(seq(i',seq(i_1',i)))$ translates into having $t \in \sigma_d(i') \globalWeakSeq \sigma_d(i_1') \globalWeakSeq \sigma_d(i)$. Then:
\begin{itemize}
    \item If $a$ is taken from the first iteration of the loop, then, given that $\epsilon \in \sigma_d(i')$ (Lem.\ref{lem:sem_de_pruned}) we have $t \in \{\epsilon\} \globalWeakSeq \{t_1'\} \globalWeakSeq \sigma_d(i)$ with $t_1'$ a continuation of the first iteration s.t. $t_1 = a.t_1' \in \sigma_d(i_1)$.
    \item If $a$ is taken from the second iteration of the loop, let us consider $t_1 \in \sigma_d(i_1)$ the first iteration and $t_2 = a.t_2' \in \sigma_d(i_1)$ the second one (from which $a$ is taken and hence $t_2' \in \sigma_d(i_1')$). We then have $t \in \{t_1\} \globalWeakSeq \{t_2'\} \globalWeakSeq \sigma_d(i)$ and the condition $\neg(t_1 \doubleVerticalTimes \theta(a))$. This condition implies, as per Lem.\ref{lem:sem_de_pruned} that $\{t_1\} \subset \sigma_d(i')$.
    \item The reasoning is the same when $a$ is taken from later instances. Let us indeed consider $a.t \in \{t_1\} \globalWeakSeq \cdots \globalWeakSeq \{t_{n-1}\} \globalWeakSeq \{t_n'\} \globalWeakSeq \sigma_d(i)$. We then have $\{t_1\} \globalWeakSeq \cdots \globalWeakSeq \{t_{n-1}\} \subset \sigma_d(i')$ because $i'$ is either a loop (and therefore absorbing) or the empty interaction (in that degenerate case all the $t_j$ are $\epsilon$) and hence the rules indeed allows $a$ to be taken from any iteration.
\end{itemize}

\begin{definition}[Operational semantics]\label{def:operational_semantics}
$\sigma_o : \mathbb{I}_\Omega \rightarrow \mathcal{P}(\mathbb{T}_\Omega)$ is s.t.:

{
\centering
\begin{minipage}{4cm}
\begin{prooftree}
\AxiomC{$i \downarrow$}
\UnaryInfC{$\epsilon \in \sigma_o(i)$}
\end{prooftree}
\end{minipage}
\begin{minipage}{4cm}
\begin{prooftree}
\AxiomC{$t \in \sigma_o(i')$}
\AxiomC{$i \xrightarrow{a} i'$}
\BinaryInfC{$a.t \in \sigma_o(i)$}
\end{prooftree}
\end{minipage}\\
}

\end{definition}

This execution relation then constitutes the small steps of the operational semantics $\sigma_o$ that we define in Def.\ref{def:operational_semantics}.

\subsection{Illustrative example\label{subsec:illustrative_example}}

\begin{figure}[ht]
    \centering
\begingroup 
\def\arraystretch{3}
\setcellgapes{1.5ex}\makegapedcells
\begin{tabular}{|c|c}
\cline{1-1}
\makecell{
\rotatebox{90}{$seq(i,i)$}
}
&
\makecell[l]{
\scalebox{.75}{
\begin{tikzpicture}[every node/.style = {shape=rectangle, align=center}]
\node[draw,shape=rectangle] (i1) at (0,0) {\includegraphics[scale=.35]{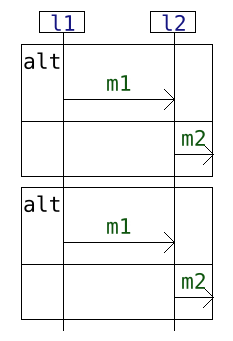}};
\node (transition1) at (2.5,0) {\Large $\xrightarrow{\hlf{l_1}!\hms{m_1}}$};
\node[draw,shape=rectangle] (i2) at (4.25,0) {\includegraphics[scale=.35]{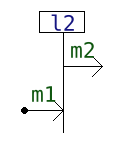}};
\node (transition2) at (6,0) {\Large $\xrightarrow{\hlf{l_2}!\hms{m_2}}$};
\node[draw,shape=rectangle] (i3) at (7.75,0) {\includegraphics[scale=.35]{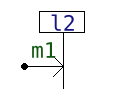}};
\node (transition3) at (9.5,0) {\Large $\xrightarrow{\hlf{l_2}?\hms{m_1}}$};
\node[draw,shape=rectangle,minimum width=1cm,minimum height=1cm] (i4) at (11,0) {};
\end{tikzpicture}
}
}
\\
\cline{1-1}
\makecell{
\rotatebox{90}{$loop_H(i)$}
}
&
\makecell[l]{
\scalebox{.75}{
\begin{tikzpicture}[every node/.style = {shape=rectangle, align=center}]
\node[draw,shape=rectangle] (i1) at (0,0) {\includegraphics[scale=.35]{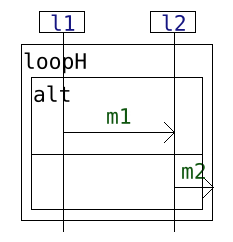}};
\node (transition1) at (2.5,0) {\Large $\xrightarrow{\hlf{l_1}!\hms{m_1}}$};
\node[draw,shape=rectangle] (i2) at (5,0) {\includegraphics[scale=.35]{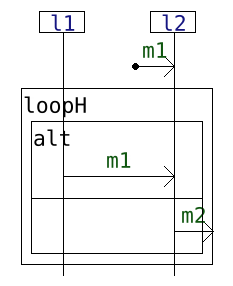}};
\end{tikzpicture}
}
}
\\
\cline{1-1}
\makecell{
\rotatebox{90}{$loop_S(i)$}
}
&
\makecell[l]{
\scalebox{.75}{
\begin{tikzpicture}[every node/.style = {shape=rectangle, align=center}]
\node[draw,shape=rectangle] (i1) at (0,0) {\includegraphics[scale=.35]{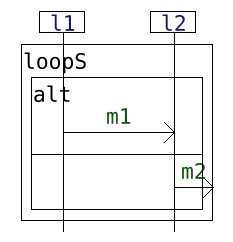}};
\node (transition1) at (2.15,0) {\Large $\xrightarrow{\hlf{l_1}!\hms{m_1}}$};
\node[draw,shape=rectangle] (i2) at (4.3,0) {\includegraphics[scale=.35]{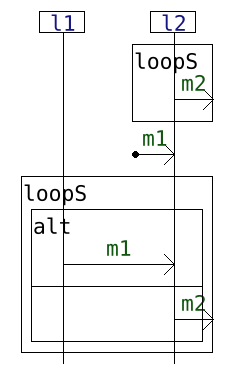}};
\node (transition2) at (6.45,0) {\Large $\xrightarrow{\hlf{l_2}!\hms{m_2}}$};
\node[draw,shape=rectangle] (i3) at (8.6,0) {\includegraphics[scale=.35]{figures/6/op_loopS_ex_after.png}};
\node (transition3) at (10.75,0) {\Large $\xrightarrow{\hlf{l_2}?\hms{m_1}}$};
\node[draw,shape=rectangle] (i4) at (12.9,0) {\includegraphics[scale=.35]{figures/6/op_loopS_ex_before.png}};
\end{tikzpicture}
}
}
\\
\cline{1-1}
\end{tabular}
\endgroup 
    \caption{Illustration of the operational semantics \& of the counter-example from Sec.\ref{sec:semantic_domain}}
    \label{fig:counter_example_with_operational_semantics}
\end{figure}

On Fig.\ref{fig:counter_example_with_operational_semantics} we illustrate both the operational semantics and the counter-example from Sec.\ref{sec:semantic_domain} showcasing the difference between $loop_H$ and $loop_S$. We consider repetitions of $i = alt(strict(l_1!m_1,l_2?m_1),l_2!m_2)$. On the first row we illustrate the construction of a trace accepted by $seq(i,i)$ where $i$ is repeated twice using weak sequencing.
Here, the second occurrence of action $l_1!m_1$ (at the bottom) is immediately executable because, with pruning, we can force the choice of the right branch of the first alternative which evades $l_1$. From that point on the execution of the remainder is straightforward, and, in total, we can conclude that the trace $t = l_1!m_1.l_2!m_2.l_2?m_1$ is expressed by $seq(i,i)$.
Now, if we were to try to reproduce this behavior using $loop_H$ to repeat $i$ (i.e. with $loop_H(i)$) we would get what is illustrated on the second row of Fig.\ref{fig:counter_example_with_operational_semantics}. We can manage to execute the first action $l_1!m_1$ but from that point, the second action of $t$ which is $l_2!m_2$ is not executable. Indeed, the presence of $l_2?m_1$ at the top of the diagram prevents it to be executed.
As we have indicated earlier (and by definition) $loop_H$ is associated to the weak HF-closure $~^{\globalWeakSeq^\Lsh *}$ and not to the K-closure $~^{\globalWeakSeq *}$. It is therefore expected that $t$ could not be accepted by $loop_H(i)$ in this example.
However, as illustrated on the third row of Fig.\ref{fig:counter_example_with_operational_semantics}, $loop_S(i)$ can recognize $t$. The addition of the pruned version of the loop allows one to delay the determination of the instance as part of which the initial $l_1!m_1$ is executed. In this case, the pruned loop is instanciated once, which means that $l_1!m_1$ was part of the second instance.

\section{Proving the equivalence of both semantics\label{sec:proof_equivalence}}

In the following we formally prove the equivalence of $\sigma_o$ and $\sigma_d$. Details of the proofs are available in the appendix and a formalisation using Coq is available in \cite{coq_hibou_label_semantics_equivalence}.

Let us at first prove that for any interaction $i$ we have $\sigma_o(i) \subset \sigma_d(i)$.
The first step to do so is to characterize the execution relation "$\rightarrow$" w.r.t. $\sigma_d$, which we do in Lem.\ref{lem:sem_de_execute1}.

\begin{lemma}[Property 1 of $\rightarrow$ w.r.t. $\sigma_d$]\label{lem:sem_de_execute1}
For any $a \in \mathbb{A}_\Omega$, $t \in \mathbb{T}_\Omega$ and $i$ and $i'$ from $\mathbb{I}_\Omega$:
\[
\left(
\begin{array}{c}
(i \xrightarrow{a} i')
\wedge 
(t \in \sigma_d(i'))
\end{array}
\right)
\Rightarrow 
(a.t \in \sigma_d(i))
\]
\end{lemma}

\begin{proof}
By induction on the 13 cases that makes the hypothesis $i \xrightarrow{a} i'$ possible.
\end{proof}

We then remark that Lem.\ref{lem:sem_de_execute1} and Lem.\ref{lem:sem_de_terminates} state that the $\sigma_d$ semantics accepts the same two construction rules (that for the empty trace $\epsilon$ and that for non empty traces of the form $a.t$) as those that define $\sigma_o$ inductively. As a result any trace that might be accepted according to $\sigma_o$ must also be accepted according to $\sigma_d$. This implies the inclusion from Th.\ref{th:sem_op_included_in_sem_de}.

\begin{theorem}[Inclusion of $\sigma_o$ in $\sigma_d$]\label{th:sem_op_included_in_sem_de}
For any $i \in \mathbb{I}_\Omega$ we have $\sigma_o(i) \subset \sigma_d(i)$
\end{theorem}

\begin{proof}
By induction on a member trace $t$.
\end{proof}

Let us now prove the reciprocate, i.e. that for any interaction $i$, $\sigma_d(i) \subset \sigma_o(i)$.
As in the previous case, we provide, with Lem.\ref{lem:sem_de_execute2}, a characterization of "$\rightarrow$" w.r.t. $\sigma_d$. Lem.\ref{lem:sem_de_execute2} is, in a certain manner, the reciprocate of Lem.\ref{lem:sem_de_execute1}.

\begin{lemma}[Property 2 of $\rightarrow$ w.r.t. $\sigma_d$]\label{lem:sem_de_execute2}
For any $a \in \mathbb{A}_\Omega$, $t \in \mathbb{T}_\Omega$ and $i \in \mathbb{I}_\Omega$:
\[
(a.t \in \sigma_d(i))
\Rightarrow
\left(
\exists~i' \in \mathbb{I}_\Omega,~
\begin{array}{c}
(i \xrightarrow{a} i')
\wedge 
(t \in \sigma_d(i'))
\end{array}
\right)
\]
\end{lemma}

\begin{proof}
By induction on the term structure of interactions.
\end{proof}

Thanks to Lem.\ref{lem:sem_de_execute2} and the characterization from Lem.\ref{lem:sem_de_terminates} we conclude with Th.\ref{th:sem_de_included_in_sem_op}.

\begin{theorem}[Inclusion of $\sigma_d$ in $\sigma_o$]\label{th:sem_de_included_in_sem_op}
For any $i \in \mathbb{I}_\Omega$ we have $\sigma_d(i) \subset \sigma_o(i)$
\end{theorem}

\begin{proof}
By induction on a member trace $t$.
\end{proof}

We have therefore proven both inclusion and can conclude that the operational semantics $\sigma_o$ is indeed equivalent to the denotational-style semantics $\sigma_d$.

\section{Related works\label{sec:related_works}}

The main inspiration for the syntax \& denotational semantics of our IL is \cite{UML_interactions_meet_state_machines_an_institutional_approach}. In \cite{UML_interactions_meet_state_machines_an_institutional_approach}, a semantics for interactions is formulated using operators on sets of traces. However, loops are not handled.
In \cite{operational_semantics_for_msc}, an operational semantics for MSC is presented. Similarities between \cite{operational_semantics_for_msc} and our work notably include the use of pruning which, in \cite{operational_semantics_for_msc}, relates to a "permission relation".
However we have some major distinctions concerning which constructors are defined and how they are handled by rules of the semantics. In \cite{operational_semantics_for_msc} loops are not handled and there is no $strict$ constructor. Indeed, causal relations between actions occurring on different lifelines (e.g. emission-reception) are handled by maps that are updated during execution. In \cite{the_many_meanings_of_uml2_sd_a_survey} a survey of formal semantics associated to UML-SDs is proposed. It is notable that UML-SDs are described semi-formally in the norm \cite{UMLNorm}. This allows for a rich language with operators such as $assert$ or $negate$ \cite{assert_and_negate_revisited_modal_semantics_for_UML_sequence_diagrams} which are not covered in our IL. However a full formalisation proves difficult, as explained in \cite{the_many_meanings_of_uml2_sd_a_survey,assert_and_negate_revisited_modal_semantics_for_UML_sequence_diagrams}.
As mentioned in the introduction, most approaches rely on translations towards other formalisms \cite{compositional_semantics_for_UML2_sequence_diagrams_using_petri_nets} or consist in denotational semantics \cite{semantics_of_interactions_in_uml_2_0} that are most often based on partial order sets. The extent to which UML-SDs are formalised may vary \cite{the_many_meanings_of_uml2_sd_a_survey}; some works formalize loops \cite{required_behavior_of_sequence_diagrams_semantics_and_conformance}, others do not \cite{UML_interactions_meet_state_machines_an_institutional_approach}, and some only allow finitely many iterations \cite{semantics_of_interactions_in_uml_2_0}. In all cases where there are loops, only one loop operator is proposed and may correspond to either $loop_H$ or $loop_S$.

Earlier works of ours \cite{revisiting_semantics_of_interactions_for_trace_validity_analysis,a_small_step_approach_to_multi_trace_checking_against_interactions}
focused on the application of the operational semantics for formal verification. In particular we were interested in the static analysis of traces and multi-traces (sets of locally defined traces) against interaction specifications (the membership problem). In those works we have used an "algorithmicized" version of the operational semantics. The inductive predicates for $\downarrow$, $\evadesLf$ and the $\isPruneBase$ and $\rightarrow$ relations were presented in functional style and we separated concerns between the determination of which actions are immediately executable (uniquely identified by their positions) and the execution of said actions. Novel contributions in this paper w.r.t. \cite{revisiting_semantics_of_interactions_for_trace_validity_analysis,a_small_step_approach_to_multi_trace_checking_against_interactions} consist in the distinction of $loop_S$ \& $loop_H$, the formulation of the operational semantics in the style of process algebras and the proof of equivalence w.r.t. a denotational semantics (a Coq proof is available in \cite{coq_hibou_label_semantics_equivalence}).

\section{Conclusion \& further work\label{sec:conclusion}}

In this paper we defined an IL for specifying the behavior of distributed and concurrent systems. This language includes weak and strict sequencing, parallel \& alternative composition as well as four distinct loop operators to specify different kinds of repetition. We formulate the semantics of this IL: (1) in denotational-style, making use of composition \& algebraic operators and (2) in operational-style by reconstructing accepted traces via the succession of atomic executions. The equivalence of both formulations is proven (Coq proof in \cite{coq_hibou_label_semantics_equivalence}).
As demonstrated in \cite{revisiting_semantics_of_interactions_for_trace_validity_analysis,a_small_step_approach_to_multi_trace_checking_against_interactions}, the operational semantics can be further exploited in formal verification techniques such as multi-trace analysis. Other techniques may be explored in further work, such as online testing, monitoring or test generation.
In addition, we may also investigate enriching our language with some form of value passing i.e. instead of exchanging abstract messages $m \in M$ we may interpret them concretely or symbolically with typed data. This last point is notably addressed in some process calculi frameworks \cite{a_symbolic_approach_to_value_passing_processes}.


\bibliography{biblio}

\appendix

\section{Proofs of section \ref{sec:semantic_domain}}

\begin{lemma*}
[\textbf{Lem.\ref{lem:strict_and_interleaving_kleene_closure_operational_charac_on_traces}}]
For any $\diamond \in \{\globalStrictSeq,~\globalInterleaving\}$, $T \in \mathcal{P}(\mathbb{T}_\Omega)$, $t$ in $\mathbb{T}_\Omega$ and $a \in \mathbb{A}_\Omega$ we have:
\[
(a.t \in T^{\diamond *}) \Rightarrow 
\left(
\exists~ t' \in \mathbb{T}_\Omega
\text{ s.t. }
\left\{
\begin{array}{l}
(a.t' \in T)\\
\wedge~(t \in (\{t'\} \diamond T^{\diamond *}))
\end{array}
\right.
\right)
\]
\end{lemma*}

\begin{proof}
By definition of the K-closure, $a.t \in T^{\diamond *}$ implies the existence of $j \geq 0$ s.t. $a.t \in T^{\diamond j}$. We can then reason by induction on the power $j$:
\begin{itemize}
    \item we cannot have $j = 0$ because $T^{\diamond 0} = \{\epsilon\}$
    \item if $j = 1$ then $a.t \in T^{\diamond 1} = T$ and $t \in \{t\} \diamond \{\epsilon\} \subset \{t\} \diamond T^{\diamond *}$ therefore the property holds
    \item if $j > 1$ the fact that $a.t \in T^{\diamond j} = T \diamond T^{\diamond (j-1)}$ implies the existence of $t'' \in T$ s.t. $a.t \in \{t''\} \diamond T^{\diamond (j-1)}$ then:
    \begin{itemize}
        \item if $\diamond = \globalStrictSeq$ this implies that:
        \begin{itemize}
            \item either $t''$ is of the form $a.t'$ and $t \in \{ t' \} \globalStrictSeq T^{\globalStrictSeq (j-1)}$ and therefore $t \in \{t'\} \globalStrictSeq T^{\globalStrictSeq *}$
            \item or $t'' = \epsilon$ and $a.t \in T^{\globalStrictSeq (j-1)}$ and we can use the induction hypothesis to conclude
        \end{itemize}
        \item if $\diamond = \globalInterleaving$ this implies that:
        \begin{itemize}
            \item either $t''$ is of the form $a.t'$ and $t \in \{ t' \} \globalInterleaving T^{\globalInterleaving (j-1)}$ and therefore $t \in \{t'\} \globalInterleaving T^{\globalInterleaving *}$
            \item or there exists a certain $t'''$ s.t. we have $a.t''' \in T^{\globalInterleaving (j-1)}$ and $t \in \{t''\} \globalInterleaving \{t'''\}$. Let us then remark that given that we have $a.t''' \in T^{\globalInterleaving (j-1)}$ we can apply the induction hypothesis to reveal $t'$ such that $a.t' \in T$ and $t''' \in \{t'\} \globalInterleaving T^{\globalInterleaving *}$. We can then use the associativity and commutativity of $\globalInterleaving$ as follows:
            \[
                \begin{array}{lclr}
                t \in \{t''\} \globalInterleaving (\{t'\} \globalInterleaving T^{\globalInterleaving *})
                &
                \Rightarrow
                &
                t \in \{t''\} \globalInterleaving (T^{\globalInterleaving *} \globalInterleaving \{t'\})
                &
                \text{ commutativity}
                \\
                &
                \Rightarrow
                &
                t \in (\{t''\} \globalInterleaving T^{\globalInterleaving *}) \globalInterleaving \{t'\}
                &
                \text{ associativity}
                \\
                &
                \Rightarrow
                &
                t \in T^{\globalInterleaving *} \globalInterleaving \{t'\}
                &
                \text{ property of Kleene closure}
                \\
                &
                \Rightarrow
                &
                t \in \{t'\} \globalInterleaving T^{\globalInterleaving *}
                &
                \text{ commutativity}
                \end{array}
            \]
        \end{itemize}
    \end{itemize}
\end{itemize}
\end{proof}

\begin{lemma*}[\textbf{Lem.\ref{lem:equivalence_headfirst_and_kleene_closure_for_strict_and_interleaving}}]
For any set of traces $T$:
\[
\begin{array}{ccc}
T^{\globalStrictSeq^{\Lsh} *} = T^{\globalStrictSeq *}
&
~~\text{ and }~~
&
T^{\globalInterleaving^{\Lsh} *} = T^{\globalInterleaving *}
\end{array}
\]
\end{lemma*}

\begin{proof}
Let us consider $\diamond \in \{\globalStrictSeq,~\globalInterleaving\}$.
By definition, we already have $T^{\diamond^\Lsh *} \subset T^{\diamond *}$. To prove the other inclusion let us reason by induction on a member trace:
\begin{itemize}
    \item by definition $\epsilon \in T^{\diamond^\Lsh *}$ and $\epsilon \in T^{\diamond *}$
    \item if the trace is of the form $a.t$ then if $a.t \in T^{\diamond *}$, then, as per Lem.\ref{lem:strict_and_interleaving_kleene_closure_operational_charac_on_traces} there exists a trace $t'$ such that $a.t' \in T$ and $t \in \{t'\} \diamond T^{\diamond *}$. This in turn implies that:
    \begin{itemize}
        \item given that action $a$ is taken from $a.t'$, we have $a.t \in \{a.t'\} \diamond^\Lsh T^{\diamond *}$
        \item and there exists $t'' \in T^{\diamond *}$ such that $t \in \{t'\} \diamond \{t''\}$. Given that $t''$ is strictly smaller than $a.t$, we can apply the induction hypothesis to obtain that $t'' \in T^{\diamond^\Lsh *}$.
    \end{itemize}
    Therefore we have that $a.t \in \{a.t'\} \diamond^\Lsh T^{\diamond^\Lsh *}$, and hence $a.t \in T^{\diamond^\Lsh *}$
\end{itemize}
\end{proof}

\section{Proofs of section \ref{sec:operational_semantics}}

\begin{lemma*}[\textbf{Lem.\ref{lem:sem_de_terminates}}]
For any $i \in \mathbb{I}_\Omega$ we have $(i \downarrow ) \Leftrightarrow (\epsilon \in \sigma_d(i))$
\end{lemma*}

\begin{proof}
Let us reason by induction on the term structure of $i$.
\begin{itemize}
    \item If $i = \varnothing$ the empty interaction, then we have both $\varnothing\downarrow$ and $\epsilon \in \sigma_d(\varnothing)$.
    \item If $ i \in \mathbb{A}_\Omega$, we have neither $i \downarrow$ nor $\epsilon \in \sigma_d(i)$.
    \item Let us now suppose that $i$ is of the form $strict(i_1,i_2)$, with $i_1$ and $i_2$ two sub-interactions that satisfy the induction hypotheses $(i_1 \downarrow ) \Leftrightarrow (\epsilon \in \sigma_d(i_1))$ and $(i_2 \downarrow ) \Leftrightarrow (\epsilon \in \sigma_d(i_2))$.
    \begin{itemize}
        \item[$\Leftarrow$] Let us suppose that $\epsilon \in \sigma_d(i)$. By definition of $\sigma_d$ for the $strict$ constructor, this implies the existence of $t_1 \in \sigma_d(i_1)$ and $t_2 \in \sigma_d(i_2)$ such that $\epsilon \in (t_1 ; t_2)$.
        This trivially implies that $t_1 = \epsilon$ and $t_2 = \epsilon$. We can therefore apply the induction hypotheses, to obtain that $i_1 \downarrow$ and $i_2 \downarrow$. This in turn means that $strict(i_1,i_2)\downarrow$.
        \item[$\Rightarrow$] Reciprocally, if $strict(i_1,i_2) \downarrow$, this means that both $i_1\downarrow$ and $i_2\downarrow$. As per the induction hypotheses, this means that $\epsilon \in \sigma_d(i_1)$ and $\epsilon \in \sigma_d(i_2)$. Therefore $\epsilon \in \sigma_d(i)$.
    \end{itemize}
    \item For interactions of the form $par(i_1,i_2)$ and $seq(i_1,i_2)$, the reasoning is the same as for the previous case.
    \item If $i$ is of the form $alt(i_1,i_2)$:
    \begin{itemize}
        \item[$\Leftarrow$] If $\epsilon \in \sigma_d(i)$ this means that either $\epsilon \in \sigma_d(i_1)$ or $\epsilon \in \sigma_d(i_2)$ or both.
        Let us suppose that it is in $\sigma_d(i_1)$ (the other cases can be treated similarly).
        As per the induction hypothesis, we have $i_1\downarrow$ which implies that $alt(i_1,i_2) \downarrow$.
        \item[$\Rightarrow$] Reciprocally, if $alt(i_1,i_2) \downarrow$, then either $i_1\downarrow$ or $i_2\downarrow$ (or both). Let us suppose we have $i_1\downarrow$. The induction hypothesis implies that $\epsilon \in \sigma_d(i_1)$ and hence $\epsilon \in \sigma_d(i)$.
    \end{itemize}
    \item If $i = loop_k(i_1)$, with $k \in \{X,H,S,P\}$ we always have $i\downarrow$ and $\epsilon \in \sigma_d(i)$.
\end{itemize}
\end{proof}

\begin{lemma*}[\textbf{Lem.\ref{lem:sem_de_evades_no_conflict}}]
For any $l \in L$ and $i \in \mathbb{I}_\Omega$, $(i \evadesLf l)
\Leftrightarrow 
(\exists~t \in \sigma_d(i), \neg(t \doubleVerticalTimes l))$
\end{lemma*}

\begin{proof}
Given $l \in L$, let us
reason by induction on the term structure of $i$.
\begin{itemize}
    \item If $i = \varnothing$, then we have $\varnothing \evadesLf l$ and $\epsilon \in \sigma_d(\varnothing)$ satisfies $\neg (\epsilon \doubleVerticalTimes l)$.
    \item If $i = a \in \mathbb{A}_\Omega$, we have:
    \begin{itemize}
        \item $a \evadesLf l$ iff $\theta(a) \neq l$
        \item and $\sigma_d(i) = \{a\}$ with trace $a$ verifying $\neg(a \doubleVerticalTimes l)$ iff $\theta(a) \neq l$
    \end{itemize}
    The two conditions are therefore equivalent.
    \item Let us now suppose that $i$ is of the form $strict(i_1,i_2)$, with $i_1$ and $i_2$ two sub-interactions that satisfy the induction hypotheses (we use the distributivity of $\doubleVerticalTimes$ over $\globalStrictSeq$):
    \begin{itemize}
        \item[$\Rightarrow$] If $i \evadesLf l$ then both $i_1 \evadesLf l$ and $i_2 \evadesLf l$. We can therefore apply the induction hypotheses, which lets us consider two traces $t_1 \in \sigma_d(i_1)$ and $t_2 \in \sigma_d(i_2)$ such that $\neg (t_1 \doubleVerticalTimes l)$ and $\neg (t_2 \doubleVerticalTimes l)$. By definition of $\sigma_d$, we have $(t_1 \globalStrictSeq t2) \subset \sigma_d(i)$. This implies that $\exists~ t \in (t_1 \globalStrictSeq t2)$ s.t.  $\neg (t \doubleVerticalTimes l)$ and this trace is in $\sigma_d(i)$.
        \item[$\Leftarrow$] Reciprocally, if $\exists~ t \in \sigma_d(strict(i_1,i_2))$ s.t. $\neg (t \doubleVerticalTimes l)$, then, by definition of $\sigma_d$, this means that there exist two traces $t_1 \in \sigma_d(i_1)$ and $t_2 \in \sigma_d(i_2)$ s.t. $t \in (t_1 \globalStrictSeq t_2)$. Then, we have $\neg (t_1 \doubleVerticalTimes l)$ and $\neg (t_2 \doubleVerticalTimes l)$. We can therefore apply the induction hypothesis which implies that $i_1 \evadesLf l$ and $i_2 \evadesLf l$, which, per the definition of $\evadesLf$, implies that $i \evadesLf l$.
    \end{itemize}
    \item For interactions of the form $par(i_1,i_2)$ and $seq(i_1,i_2)$, the reasoning is the same as for the previous case except that we reason on (respectively) the operators $\globalInterleaving$ and $\globalWeakSeq$ over both of which the conflict $\doubleVerticalTimes$ is also distributive.
    \item If $i = alt(i_1,i_2)$:
    \begin{itemize}
        \item[$\Rightarrow$] If $i \evadesLf l$ then either or both $i_1 \evadesLf l$ and $i_2 \evadesLf l$. Let us suppose we have $i_1 \evadesLf l$ (the other cases are similar). By the induction hypothesis, we have $\exists~ t_1 \in \sigma_d(i_1)$ s.t. $\neg(t_1 \doubleVerticalTimes l)$. We then simply observe that $\sigma_d(i_1) \subset \sigma_d(i)$.
        \item[$\Leftarrow$] Reciprocally, if $\exists~ t \in \sigma_d(alt(i_1,i_2))$ s.t. $\neg (t \doubleVerticalTimes l)$, this means that this trace $t$ must be found in either $\sigma_d(i_1)$ or $\sigma_d(i_2)$. Let us suppose the first case (the other is similar). We can therefore apply the induction hypothesis which implies that $i_1 \evadesLf l$. Then, by the definition of $\evadesLf$, this implies that $i \evadesLf l$.
    \end{itemize}
    \item Let us finally consider the case where $i$ is of the form $loop_k(i_1)$, with $k \in \{X,H,S,P\}$. By definition, we always have $i \evadesLf l$ and the empty trace $\epsilon \in \sigma_d(i)$ verifies $\neg(\epsilon \doubleVerticalTimes l)$.
\end{itemize}
\end{proof}

\begin{lemma*}[\textbf{Lem.\ref{lem:sem_de_pruned}}]
For any $l \in L$ and any $i$ and $i'$ from $\mathbb{I}_\Omega$:
\[
(i \isPruneOf{l} i')
\Rightarrow
(\sigma_d(i') = \{ t \in \sigma_d(i)~|~ \neg (t \doubleVerticalTimes l) \})
\]
\end{lemma*}

\begin{proof}
Given $l \in L$, let us reason by induction on the term structure of $i$.
\begin{itemize}
    \item If $i=\varnothing$ then $\varnothing \isPruneOf{l} \varnothing$ and $\sigma_d(\varnothing) = \{\epsilon\}$. $\epsilon$ having no conflict w.r.t. $l$, the property holds.
    \item If $i = a \in \mathbb{A}_\Omega$, the precondition $a \evadesLf l$ implies that $\theta(i) \neq l$. Also, we have $a \isPruneOf{l} a$ and $\sigma_d(a) = \{a\}$ has a single trace with no conflict w.r.t. $l$. Hence the property holds.
    \item Let us now suppose that $i$ is of the form $strict(i_1,i_2)$, with $i_1$ and $i_2$ two sub-interactions that satisfy induction hypotheses i.e. such that, given $i_1 \isPruneOf{l} i_1'$ and $i_2 \isPruneOf{l} i_2'$
    we have $\sigma_d(i_1') = \{ t_1 \in \sigma_d(i_1)~|~ \neg (t_1 \doubleVerticalTimes l) \}$ and $\sigma_d(i_2') = \{ t_2 \in \sigma_d(i_2)~|~ \neg (t_2 \doubleVerticalTimes l) \}$.
    Also, by definition of pruning, we have $i \isPruneOf{l} strict(i_1',i_2')$ and let us denote it by $i' = strict(i_1',i_2')$ for short.
    \begin{itemize}
        \item[$\subset$] If $t \in \sigma_d(i')$ then there exist $t_1 \in \sigma_d(i_1')$ and $t_2 \in \sigma_d(i_2')$ such that $t \in (t_1 ; t_2)$. By the induction hypothesis, we have $t_1 \in \sigma_d(i_1)$ and $\neg (t_1 \doubleVerticalTimes) l$ and $t_2 \in \sigma_d(i_2)$ and $\neg (t_2 \doubleVerticalTimes) l$. Therefore $t \in \sigma_d(i_1) ; \sigma_d(i_2) = \sigma_d(i)$, and $\neg (t \doubleVerticalTimes l)$.
        \item[$\supset$] If $t \in \sigma_d(i)$ is s.t. $\neg (t \doubleVerticalTimes l)$ then this implies the existence of $t_1 \in \sigma_d(i_1)$ and $t_2 \in \sigma_d(i_2)$ s.t. $t \in (t_1 \globalStrictSeq t_2)$. The fact that $\neg (t \doubleVerticalTimes l)$ implies that both $\neg (t_1 \doubleVerticalTimes l)$ and $\neg (t_2 \doubleVerticalTimes l)$. According to the induction hypothesis, this means that $t_1 \in \sigma_d(i_1')$ and $t_2 \in \sigma_d(i_2')$. Therefore $(t_1 \globalStrictSeq t_2) \subset \sigma_d(i')$ and hence $t \in \sigma_d(i')$.
    \end{itemize}
    \item For interactions of the form $par(i_1,i_2)$ and $seq(i_1,i_2)$, the reasoning is the same as for the previous case except that we reason on (respectively) the operators $\globalInterleaving$ and $\globalWeakSeq$ over both of which the conflict $\doubleVerticalTimes$ is also distributive.
    \item Let us now suppose that $i$ is of the form $loop_k(i_1)$ with $k \in \{X,H,S,P\}$ and let us note $\diamond$ the corresponding operator on sets of traces, i.e. $\diamond = \globalStrictSeq$ if $k=X$, $\diamond = \globalWeakSeq^\Lsh$ if $k=H$, $\diamond = \globalWeakSeq$ if $k=S$ and $\diamond = \globalInterleaving$ if $k=P$. We then have:
    \begin{itemize}
        \item if $i_1 \collidesLf l$, then $loop_k(i_1) \isPruneOf{l} \varnothing$. As per the reciprocate of Lem.\ref{lem:sem_de_evades_no_conflict}, if $i_1$ does not avoid $l$, this means that all the traces from $\sigma_d(i_1)$ have conflicts with $l$. This means that all the traces obtained from merging traces from $i_1$ have conflicts with $l$. Therefore the empty trace $\epsilon$ is the only trace from $\sigma_d(loop_k(i_1))$ which has no conflict with $l$. Given that $\sigma_d(\varnothing) = \{\epsilon\}$, the property holds.
        \item if $i_1 \evadesLf l$, then there exists a unique $i_1'$ such that $i_1 \isPruneOf{l} i_1'$ and we suppose the induction hypothesis $\sigma_d(i_1') = \{ t_1 \in \sigma_d(i_1)~|~ \neg (t_1 \doubleVerticalTimes l) \}$. Then we have:
\[
\begin{array}{lcl}
\sigma_d(loop_k(i_1'))
& =
& \sigma_d(i_1')^{\diamond *}\\
& =
& \{ t \in \sigma_d(i_1) ~|~ \neg (t \doubleVerticalTimes l) \}^{\diamond *}\\
& =
& \{ t \in \sigma_d(i_1)^{\diamond *} ~|~ \neg (t \doubleVerticalTimes l) \}\\
& =
& \{ t \in \sigma_d(loop_k(i_1)) ) ~|~ \neg (t \doubleVerticalTimes l) \}
\end{array}
\]
Indeed, the conflict $\doubleVerticalTimes$ distributes over $\diamond$ and hence any trace obtained from merging traces from $i_1$ has no conflict w.r.t. $l$ iff it is obtained from merging traces from $i_1$ that all have no conflict with $l$. Therefore, traces that have no conflict w.r.t. $l$ are exactly those that are obtained from merging traces with no conflicts w.r.t. $l$. Those traces are those from $loop_k(i_1')$ as per the induction hypothesis. Therefore the property holds.
    \end{itemize}
\end{itemize}
\end{proof}

\section{Proofs of section \ref{sec:proof_equivalence}}

\begin{lemma*}[\textbf{Lem.\ref{lem:sem_de_execute1}}]
For any $a \in \mathbb{A}_\Omega$, $t \in \mathbb{T}_\Omega$ and $i$ and $i'$ from $\mathbb{I}_\Omega$:
\[
\left(
\begin{array}{c}
(i \xrightarrow{a} i')
\wedge 
(t \in \sigma_d(i'))
\end{array}
\right)
\Rightarrow 
(a.t \in \sigma_d(i))
\]
\end{lemma*}

\begin{proof}
Given $i$ and $i'$ in $\mathbb{I}_\Omega$, $a$ in $\mathbb{A}_\Omega$ and $t \in \mathbb{T}_\Omega$, let us suppose that $i \xrightarrow{a} i'$ and that $t \in \sigma_d(i')$. In order to prove $a.t \in \sigma_d(i)$ let us reason by induction on the cases that makes the hypothesis $i \xrightarrow{a} i'$ possible. There are 13 such cases, as per the 13 rules from Def.\ref{def:execution_relation}:
\begin{enumerate}
    \item when executing an atomic action, we have $i \in \mathbb{A}_\Omega$ and $i' = \varnothing$. Then $\sigma_d(i) = \{i\}$ and $\sigma_d(\varnothing) = \{\epsilon\}$. The property $i.\epsilon = i \in \sigma_d(i)$ holds.
    \item when executing an action on the left of an alternative, we have $i$ of the form $alt(i_1,i_2)$, and $i' = i_1'$ such that $i_1 \xrightarrow{a} i_1'$. As a result, $i_1 \xrightarrow{a} i_1'$ and $t \in \sigma_d(i_1')$. We can therefore apply the induction hypothesis on sub-interaction $i_1$ to obtain that $a.t \in \sigma_d(i_1)$. Given that $\sigma_d(i_1) \subset \sigma_d(i)$, the property holds.
    \item the case for executing an action on the right of an alternative can be treated similarly
    \item when executing an action on the left of a $par$, we have $i$ of the form $par(i_1,i_2)$, and $i' = par(i_1',i_2)$ such that $i_1 \xrightarrow{a} i_1'$ and $t \in \sigma_d(par(i_1',i_2))$. By definition of $\sigma_d$, we have that there exist $(t_1',t_2) \in \sigma_d(i_1') \times \sigma_d(i_2)$ s.t. $t \in (t_1' \globalInterleaving t_2)$. Therefore we have $i_1 \xrightarrow{a} i_1'$ and $t_1' \in \sigma_d(i_1')$. Hence we can apply the induction hypothesis on sub-interaction $i_1$ to obtain that $a.t_1' \in \sigma_d(i_1)$. Given that $\sigma_d(par(i_1,i_2))$ is the union of all the $(t_\alpha \globalInterleaving t_\beta)$ with $t_\alpha$ and $t_\beta$ traces from $i_1$ and $i_2$, we have that $(a.t_1' \globalInterleaving t_2) \subset \sigma_d(i)$. In particular, we know that $t \in (t_1' \globalInterleaving t_2)$, so, by definition of the $\globalInterleaving$ operator, we have that $a.t \in (a.t_1' \globalInterleaving t_2)$. Therefore the property holds.
    \item the case for executing an action on the right of a $par$ can be treated similarly
    \item executing an action on the left of a $strict$ can be treated like \textbf{4.}
    \item when executing an action on the right of a $strict$, we have $i$ of the form $strict(i_1,i_2)$, and $i' = i_2'$ such that $i_2 \xrightarrow{a} i_2'$ with the added hypothesis that $i_1\downarrow$. Given that $t \in \sigma_d(i_2')$ and $i_2 \xrightarrow{a} i_2'$, we can apply the induction hypothesis on sub-interaction $i_2$ to obtain that $a.t \in \sigma_d(i_2)$. Given that $\sigma_d(strict(i_1,i_2))$ includes $\sigma_d(i_2)$ when $i_1\downarrow$, and given that we know $i_1\downarrow$ to be true, the property holds.
    \item executing an action on the left of a $seq$ can be treated like \textbf{4.}
    \item when executing an action on the right of a $seq$, we have $i$ of the form $seq(i_1,i_2)$, and $i' = seq(i_1',i_2')$ such that $i_1 \isPruneOf{\theta(a)} i_1'$ and $i_2 \xrightarrow{a} i_2'$ with the added hypothesis that $i_1 \evadesLf \theta(a)$ implied by the fact that $i_1$ prunes into $i_1'$. Given that $t \in \sigma_d(i')$, there exists $t_1 \in \sigma_d(i_1')$ and $t_2 \in \sigma_d(i_2')$ s.t. $t \in (t_1 \globalWeakSeq t_2)$. We then remark that:
    \begin{itemize}
        \item $i_2 \xrightarrow{a} i_2'$ and $t_2 \in \sigma_d(i_2')$. Hence we can apply the induction hypothesis on sub-interaction $i_2$ to obtain that $a.t_2 \in \sigma_d(i_2)$.
        \item the fact that $t_1 \in \sigma_d(i_1')$ implies, as per Lem.\ref{lem:sem_de_pruned} that $\neg(t_1 \doubleVerticalTimes \theta(a))$. As a result, by definition of the weak sequencing operator, $(t_1 \globalWeakSeq a.t_2)$ includes $a.t$.
    \end{itemize}
    Finally, given that $\sigma_d(i)$ includes all $(t_\alpha \globalWeakSeq t_\beta)$ s.t. $t_\alpha \in  \sigma_d(i_1)$ and $t_\beta \in \sigma_d(i_2)$, we have, in particular $(t_1 \globalWeakSeq a.t_2) \subset \sigma_d(i)$. Hence the property holds.
    \item when executing an action underneath a $loop_X$, we have $i$ of the form $loop_X(i_1)$ and $i' = strict(i_1',loop_X(i_1))$ such that $i_1 \xrightarrow{a} i_1'$. We have that $t \in \sigma_d(i')$. Therefore there exists $t_1 \in \sigma_d(i_1')$ and $t_2 \in \sigma_d(i)$ s.t. $t \in (t_1 \globalStrictSeq t_2)$. We then remark that:
    \begin{itemize}
        \item $i_1 \xrightarrow{a} i_1'$ and $t_1 \in \sigma_d(i_1')$. Hence we can apply the induction hypothesis on sub-interaction $i_1$, which implies that $a.t_1 \in \sigma_d(i_1)$.
        \item and as a result, given that $t_2 \in \sigma_d(loop_X(i_1)) = \sigma_d(i_1)^{\globalStrictSeq *}$, and $a.t_1 \in \sigma_d(i_1)$, we have, $(a.t_1 \globalStrictSeq t_2) \subset \sigma_d(i_1)^{\globalStrictSeq *}$ i.e. $(a.t_1 \globalStrictSeq t_2) \subset \sigma_d(i)$
    \end{itemize}
    Then, given that $t \in (t_1 \globalStrictSeq t_2)$, we have immediately that $a.t \in (a.t_1 \globalStrictSeq t_2)$ because it is always possible to add actions from the left.  Therefore $a.t \in \sigma_d(i)$, so the property holds.
    \item executing an action underneath a $loop_H$ can be treated like \textbf{10.}
    \item when executing an action underneath a $loop_S$, we have $i$ of the form $loop_S(i_1)$ and $i' = seq(i_0',seq(i_1',loop_S(i_1)))$ such that $i_1 \xrightarrow{a} i_1'$ and $i \isPruneOf{\theta(a)} i_0'$. Given that $t \in \sigma_d(i')$, there exists $t_0 \in \sigma_d(i_0')$, $t_1 \in \sigma_d(i_1')$ and $t_2 \in \sigma_d(i)$ s.t. $t \in (t_0 \globalWeakSeq t_1 \globalWeakSeq t_2)$. We then remark that:
    \begin{itemize}
        \item Given that $i \isPruneOf{\theta(a)} i_0'$ we have that:
        \begin{itemize}
            \item $\sigma_d(i_0') \subset \sigma_d(i)$ and, given that $\sigma_d(i) = \sigma_d(i_1)^{\globalWeakSeq *}$, is a weak Kleene closure we have that $\sigma_d(i) \globalWeakSeq \sigma_d(i) \subset \sigma_d(i)$ and therefore $\sigma_d(i_0') \globalWeakSeq \sigma_d(i) \subset \sigma_d(i)$
            \item and, given that $t_0 \in \sigma_d(i_0')$, we have as per Lem.\ref{lem:sem_de_pruned} that $\neg(t_0 \doubleVerticalTimes \theta(a))$. Therefore, for any $t_\beta$ and any $t_\alpha \in (t_0 \globalWeakSeq t_\beta)$ we have $a.t_\alpha \in (t_0 \globalWeakSeq a.t_\beta)$ given that we can take action $a$, which have no conflict w.r.t $t_0$, on the right side
        \end{itemize}
        \item We have $i_1 \xrightarrow{a} i_1'$ and $t_1 \in \sigma_d(i_1')$. Hence we can apply the induction hypothesis on sub-interaction $i_1$, which implies that $a.t_1 \in \sigma_d(i_1)$
        \item Given that $t_2 \in \sigma_d(loop_S(i_1)) = \sigma_d(i_1)^{\globalWeakSeq *}$, and $a.t_1 \in \sigma_d(i_1)$, we have, $(a.t_1 \globalWeakSeq t_2) \subset \sigma_d(i_1)^{\globalWeakSeq *}$ i.e. $(a.t_1 \globalWeakSeq t_2) \subset \sigma_d(i)$
    \end{itemize}
     Finally, given $a.t \in (t_0 \globalWeakSeq a.t_1 \globalWeakSeq t_2) \subset (\sigma_d(i_0') \globalWeakSeq \sigma_d(i)) \subset \sigma_d(i)$, the property holds.
    \item executing an action underneath a $loop_P$ can be treated like \textbf{10.}
\end{enumerate}
\end{proof}

\begin{theorem*}[\textbf{Th.\ref{th:sem_op_included_in_sem_de}}]
For any $i \in \mathbb{I}_\Omega$ we have $\sigma_o(i) \subset \sigma_d(i)$
\end{theorem*}

\begin{proof}
Let us consider $i \in \mathbb{I}_\Omega$ and $t \in \sigma_o(i)$ and let us reason by induction on the trace $t$.
\begin{itemize}
    \item If $t = \epsilon$, then, as per the definition of $\sigma_o$, this means that $i\downarrow$. Then as per Lem.\ref{lem:sem_de_terminates}, this means that $\epsilon \in \sigma_d(i)$.
    \item If $t = a.t'$ then, by definition of $\sigma_o$, $a.t' \in \sigma_o(i)$ iff $\exists~ i' \in \mathbb{I}_\Omega$ s.t. $i \xrightarrow{a} i'$ and $t' \in \sigma_o(i')$. Let us therefore consider such an interaction $i'$. By the induction hypothesis on trace $t'$, we have $(t' \in \sigma_o(i')) \Rightarrow (t' \in \sigma_d(i'))$. As a result we have $i \xrightarrow{a} i'$ and $t' \in \sigma_d(i')$. We can therefore apply Lem.\ref{lem:sem_de_execute1} to conclude that $a.t' \in \sigma_d(i)$. Hence the property holds.
\end{itemize}
\end{proof}

\begin{lemma}[\textbf{Lem.\ref{lem:sem_de_execute2}}]
For any $a \in \mathbb{A}_\Omega$, $t \in \mathbb{T}_\Omega$ and $i \in \mathbb{I}_\Omega$:
\[
(a.t \in \sigma_d(i))
\Rightarrow
\left(
\exists~i' \in \mathbb{I}_\Omega,~
\begin{array}{c}
(i \xrightarrow{a} i')
\wedge 
(t \in \sigma_d(i'))
\end{array}
\right)
\]
\end{lemma}

\begin{proof}
Let us consider $i \in \mathbb{I}_\Omega$, $a \in \mathbb{A}_\Omega$ and $t \in \mathbb{T}_\Omega$. Let us suppose that $a.t \in \sigma_d(i)$ and let us reason by induction on the term structure of $i$.
\begin{itemize}
    \item we cannot have $i = \varnothing$ because it contradicts $a.t \in \sigma_d(i)$
    \item if $i \in \mathbb{A}_\Omega$ then $a.t \in \sigma_d(i)$ implies that $i=a$ and $t = \epsilon$. We then have the existence of $i' = \varnothing$ which indeed satisfies that $a \xrightarrow{a} \varnothing$ and $\epsilon \in \sigma_d(\varnothing)$
    \item if $i$ is of the form $alt(i_1,i_2)$ then $a.t \in \sigma_d(i)$ implies either $a.t \in \sigma_d(i_1)$ or $a.t \in \sigma_d(i_2)$. Let us suppose it is the first case (the second is identical). Then, we can apply the induction hypothesis on sub-interaction $i_1$, which reveals the existence of $i_1'$ such that $i_1 \xrightarrow{a} i_1'$ and $t \in \sigma_d(i_1')$. By definition of the execution relation "$\rightarrow$", this implies that $alt(i_1,i_2) \xrightarrow{a} i_1'$. As a result, we have identified $i'=i_1'$ which satisfies the property.
    \item if $i$ is of the form $par(i_1,i_2)$ then $a.t \in \sigma_d(i)$ implies the existence of traces $t_1$ and $t_2$ such that $t_1 \in \sigma_d(i_1)$ and $t_2 \in \sigma_d(i_2)$ and $a.t \in (t_1 \globalInterleaving t_2)$. This then implies:
    \begin{itemize}
        \item either that $t_1$ is of the form $a.t_1'$ and $t \in (t_1' \globalInterleaving t_2)$
        \item or $t_2$ is of the form $a.t_2'$ and $t \in (t_1 \globalInterleaving t_2')$
    \end{itemize}
    As both case can be treated identically, let us suppose it is the first case. Given that we have $a.t_1' \in \sigma_d(i_1)$, we can apply the induction hypothesis on sub-interaction $i_1$, which reveals the existence of $i_1'$ such that $i_1 \xrightarrow{a} i_1'$ and $t_1' \in \sigma_d(i_1')$. By definition of the execution relation "$\rightarrow$", this implies that $par(i_1,i_2) \xrightarrow{a} par(i_1',i_2)$. By definition of $\sigma_d$, given that $t_1' \in \sigma_d(i_1')$ and $t_2 \in \sigma_d(i_2)$, we have $(t_1' \globalInterleaving t_2) \subset \sigma_d(par(i_1',i_2))$. Then, given that $t \in (t_1' \globalInterleaving t_2)$, this implies that $t \in \sigma_d(par(i_1',i_2))$. We therefore have identified $i' = par(i_1',i_2)$ which satisfies the property.
    \item if $i$ is of the form $strict(i_1,i_2)$ then $a.t \in \sigma_d(i)$ implies the existence of traces $t_1$ and $t_2$ such that $t_1 \in \sigma_d(i_1)$ and $t_2 \in \sigma_d(i_2)$ and $a.t \in (t_1 \globalStrictSeq t_2)$. This then implies:
    \begin{itemize}
        \item either that $t_1$ is of the form $a.t_1'$ and $t \in (t_1' \globalStrictSeq t_2)$. In that case we can apply the induction hypothesis on sub-interaction $i_1$, which reveals the existence of $i_1'$ s.t. $i_1 \xrightarrow{a} i_1'$ and $t_1' \in \sigma_d(i_1')$. By definition of the execution relation "$\rightarrow$", this implies that $strict(i_1,i_2) \xrightarrow{a} strict(i_1',i_2)$. By definition of $\sigma_d$, given that $t_1' \in \sigma_d(i_1')$ and $t_2 \in \sigma_d(i_2)$, we have $(t_1' \globalStrictSeq t_2) \subset \sigma_d(strict(i_1',i_2))$.
        Then, given that  $t \in (t_1' \globalStrictSeq t_2)$ this implies that $t \in \sigma_d(strict(i_1',i_2))$. Hence $i' = strict(i_1',i_2)$ satisfies the property.
        \item or that $t_1 = \epsilon$ and $t_2 = a.t$. On the one hand, the fact that $t_1 = \epsilon \in \sigma_d(i_1)$ implies, as per  Lem.\ref{lem:sem_de_terminates}, that $i_1\downarrow$. On the other hand, with $t_2 = a.t \in \sigma_d(i_2)$, we can apply the induction hypothesis on sub-interaction $i_2$, which reveals the existence of $i_2'$ s.t. $i_2 \xrightarrow{a} i_2'$ and $t \in \sigma_d(i_2')$. By definition of the execution relation "$\rightarrow$", and because the precondition $i_1\downarrow$ is verified, this implies that $strict(i_1,i_2) \xrightarrow{a} i_2'$. As a result, we have identified $i' = i_2'$ which satisfies the property.
    \end{itemize}
    \item if $i$ is of the form $seq(i_1,i_2)$ then $a.t \in \sigma_d(i)$ implies the existence of traces $t_1$ and $t_2$ such that $t_1 \in \sigma_d(i_1)$ and $t_2 \in \sigma_d(i_2)$ and $a.t \in (t_1 \globalWeakSeq t_2)$. This then implies:
    \begin{itemize}
        \item either that $t_1 = a.t_1'$ and $t \in (t_1' \globalWeakSeq t_2)$. In that case we can apply the induction hypothesis on sub-interaction $i_1$, which reveals the existence of $i_1'$ s.t. $i_1 \xrightarrow{a} i_1'$ and $t_1' \in \sigma_d(i_1')$. By definition of the execution relation "$\rightarrow$", this implies that $seq(i_1,i_2) \xrightarrow{a} seq(i_1',i_2)$. By definition of $\sigma_d$, given that $t_1' \in \sigma_d(i_1')$ and $t_2 \in \sigma_d(i_2)$, we have $(t_1' \globalWeakSeq t_2) \subset \sigma_d(seq(i_1',i_2))$. Then, given that $t \in (t_1' \globalWeakSeq t_2)$, this implies that $t \in \sigma_d(seq(i_1',i_2))$. We therefore have identified $i' = seq(i_1',i_2)$ which satisfies the property.
        \item or that $\neg(t_1 \doubleVerticalTimes \theta(a))$ and that $t_2$ is of the form $a.t_2'$ with $t \in (t_1 \globalWeakSeq t_2')$. On the one hand, the fact that $t_1 \in \sigma_d(i_1)$ is such that $\neg(t_1 \doubleVerticalTimes \theta(a))$, ensures, as per Lem.\ref{lem:sem_de_evades_no_conflict}, that $i_1 \evadesLf \theta(a)$ and therefore, as per Lem.\ref{lem:pruning_existence_unicity}, that there exists a unique $i_1'$ such that $i_1 \isPruneOf{\theta(a)} i_1'$.
        On the other hand, with $t_2 = a.t_2' \in \sigma_d(i_2)$, we can apply the induction hypothesis on sub-interaction $i_2$, which reveals the existence of $i_2'$ s.t. $i_2 \xrightarrow{a} i_2'$ and $t_2' \in \sigma_d(i_2')$.
        By definition of the execution relation "$\rightarrow$", this implies that $seq(i_1,i_2) \xrightarrow{a} seq(i_1',i_2')$.
        Given that $t_1 \in \sigma_d(i_1)$ is such that $\neg(t_1 \doubleVerticalTimes \theta(a))$, as per Lem.\ref{lem:sem_de_pruned}, this implies that $t_1 \in \sigma_d(i_1')$.
        By definition of $\sigma_d$, given that $t_1 \in \sigma_d(i_1')$ and $t_2' \in \sigma_d(i_2')$, we have $(t_1 \globalWeakSeq t_2') \subset \sigma_d(seq(i_1',i_2'))$.
        Then, given that $t \in (t_1 \globalWeakSeq t_2')$, this implies that $t \in \sigma_d(seq(i_1',i_2'))$.
        We therefore have identified $i' = seq(i_1',i_2')$ which satisfies the property.
    \end{itemize}
    \item if $i$ is of the form $loop_X(i_1)$ then $a.t \in \sigma_d(i)$ means that $a.t \in \sigma_d(i_1)^{\globalStrictSeq *}$ and hence, as per Lem.\ref{lem:strict_and_interleaving_kleene_closure_operational_charac_on_traces}, this implies the existence of a trace $t'$ such that $a.t' \in \sigma_d(i_1)$ and $t \in \{t'\} \globalStrictSeq \sigma_d(i_1)^{\globalStrictSeq *}$. 
    Then, the fact that $a.t' \in \sigma_d(i_1)$ allows us to apply the induction hypothesis on sub-interaction $i_1$ to reveal the existence of $i_1'$ such that $i_1 \xrightarrow{a} i_1'$ and $t' \in \sigma_d(i_1')$. By definition of the execution relation "$\rightarrow$", this implies that $loop_X(i_1) \xrightarrow{a} strict(i_1',loop_X(i_1))$. By definition of $\sigma_d$, given that $t' \in \sigma_d(i_1')$ and that $t \in \{t'\} \globalStrictSeq \sigma_d(i_1)^{\globalStrictSeq *}$, we have that $t \in \sigma_d(strict(i_1',loop_X(i_1)))$. Hence $i' = strict(i_1',loop_X(i_1))$ satisfies the property.
    \item for $i$ is of the form $loop_P(i_1)$ the proof can be handled as in the previous case
    \item if $i$ is of the form $loop_H(i_1)$ then $a.t \in \sigma_d(i)$ means that $a.t \in \sigma_d(i_1)^{\globalWeakSeq^\Lsh *}$. By definition, there exists $j > 0$ such that $a.t \in \sigma_d(i_1)^{\globalWeakSeq^\Lsh j} = \sigma_d(i_1) \globalWeakSeq^\Lsh \sigma_d(i_1)^{\globalWeakSeq^\Lsh (j-1)} \subset \sigma_d(i_1) \globalWeakSeq^\Lsh \sigma_d(i_1)^{\globalWeakSeq^\Lsh *}$.
    Because the restricted operator $\globalWeakSeq^\Lsh$ only allows to take the first action of recomposed traces from the left hand side, we can identify a trace $t'$ s.t. $a.t' \in \sigma_d(i_1)$ and $t \in \{t'\} \globalWeakSeq \sigma_d(i_1)^{\globalWeakSeq^\Lsh *}$.
    Hence, we can apply the induction hypothesis on sub-interaction $i_1$ to reveal the existence of $i_1'$ such that $i_1 \xrightarrow{a} i_1'$ and $t' \in \sigma_d(i_1')$. By definition of the execution relation "$\rightarrow$", this implies that $loop_H(i_1) \xrightarrow{a} seq(i_1',loop_H(i_1))$. By definition of $\sigma_d$, given that $t' \in \sigma_d(i_1')$ and that $t \in \{t'\} \globalWeakSeq \sigma_d(i_1)^{\globalWeakSeq^\Lsh *}$, we have that $t \in \sigma_d(seq(i_1',loop_H(i_1)))$. We therefore have identified $i' = seq(i_1',loop_H(i_1))$ which satisfies the property.
    \item if $i$ is of the form $loop_S(i_1)$ then $a.t \in \sigma_d(i)$ means that $a.t \in \sigma_d(i_1)^{\globalWeakSeq *}$. By definition there exists $n > 0$ such that $a.t \in \sigma_d(i_1)^{\globalWeakSeq n}$. As a result, we can identify traces $t_1$ through $t_n$ such that for any $j \in [1,n]$, $t_j \in \sigma_d(i_1)$ and $a.t \in \{t_1\} \globalWeakSeq \cdots \globalWeakSeq \{t_n\}$. By definition of the weak sequencing operator, action $a$ is taken from a certain $t_j$ with $j \in [1,n]$ which is therefore of the form $t_j = a.t_j'$ and we have, for any $k<j$, $\neg(t_k \doubleVerticalTimes \theta(a))$ (otherwise we could not take $a$ from $t_j$) and $t \in \{t_1\} \globalWeakSeq \cdots \globalWeakSeq \{t_{j-1}\} \globalWeakSeq \{t_j'\} \globalWeakSeq \{t_{j+1}\} \globalWeakSeq \cdots \globalWeakSeq \{t_n\}$. We can then remark the following:
    \begin{itemize}
        \item considering $i_0'$ such that $loop_S(i_1) \isPruneOf{\theta(a)} i_0'$ (which existence is guaranteed by Lem.\ref{lem:pruning_existence_unicity} given that a loop always evades any lifeline), because, for any $k<j$, we have $\neg(t_k \doubleVerticalTimes \theta(a))$ then, as per Lem.\ref{lem:sem_de_pruned}, for all $k<j$, we have $t_k \in \sigma_d(i_0')$. Then:
        \begin{itemize}
            \item If all the $t_k$ are the empty trace then $\{t_1\} \globalWeakSeq \cdots \globalWeakSeq \{t_{j-1}\} = \{\epsilon\} \subset \sigma_d(i_0')$ ($\sigma_d(i_0')$ contains at least $\epsilon$ because $loop_S(i_1)$ does)
            \item If at least one $t_k$ is not the empty trace then $i_0'$ is a non empty $loop_S$, and, because it is a $loop_S$, given that for all $k<j$, we have $t_k \in \sigma_d(i_0')$ then $\{t_1\} \globalWeakSeq \cdots \globalWeakSeq \{t_{j-1}\} \subset \sigma_d(i_0')$ (because a $loop_S$ is closed under repetition by  $\globalWeakSeq$)
        \end{itemize}
        Hence $\{t_1\} \globalWeakSeq \cdots \globalWeakSeq \{t_{j-1}\} \subset \sigma_d(i_0')$ and therefore $t \in \sigma_d(i_0') \globalWeakSeq \{t_j'\} \globalWeakSeq \{t_{j+1}\} \globalWeakSeq \cdots \globalWeakSeq \{t_n\}$
        \item given that, for any $k>j$ we have that $t_k \in \sigma_d(i)$ and because $i$ is a loop ($i = loop_S(i_1)$), then $\{t_{j+1}\} \globalWeakSeq \cdots \globalWeakSeq \{t_n\} \subset \sigma_d(i)$ and therefore $t \in \sigma_d(i_0') \globalWeakSeq \{t_j'\} \globalWeakSeq \sigma_d(i)$
        \item given that $t_j = a.t_j' \in \sigma_d(i_1)$ we can apply the induction hypothesis on sub-interaction $i_1$ to reveal the existence of $i_1'$ such that $i_1 \xrightarrow{a} i_1'$ and $t_j' \in \sigma_d(i_1')$. By definition of the execution relation "$\rightarrow$", this implies that $loop_S(i_1) \xrightarrow{a} seq(i_0',seq(i_1',loop_S(i_1)))$
        \item finally, given that we have shown that $t \in \sigma_d(i_0') \globalWeakSeq \{t_j'\} \globalWeakSeq \sigma_d(i)$ and because $t_j' \in \sigma_d(i_1')$, by definition of $\sigma_d$ we can conclude that $t \in \sigma_d(seq(i_0',seq(i_1',loop_S(i_1))))$. Therefore we have identified $i' = seq(i_0',seq(i_1',loop_S(i_1)))$ which satisfies the property.
    \end{itemize}
\end{itemize}
\end{proof}

\begin{theorem*}[\textbf{Th.\ref{th:sem_de_included_in_sem_op}}]
For any $i \in \mathbb{I}_\Omega$ we have $\sigma_o(i) \supset \sigma_d(i)$
\end{theorem*}

\begin{proof}
Let us consider $i \in \mathbb{I}_\Omega$ and $t \in \sigma_d(i)$ and let us reason by induction on the trace $t$.
\begin{itemize}
    \item If $t = \epsilon$, the fact that $t = \epsilon \in \sigma_d(i)$ implies, as per Lem.\ref{lem:sem_de_terminates}, that $i\downarrow$. Then, by definition of $\sigma_o$, this means that $\epsilon \in \sigma_o(i)$.
    \item If $t = a.t'$ then, the fact that $a.t' \in \sigma_d(i)$ implies, as per Lem.\ref{lem:sem_de_execute2}, that there exists an interaction $i'$ such that $i \xrightarrow{a} i'$ and $t' \in \sigma_d(i')$. Let us therefore consider such an interaction $i'$. By the induction hypothesis on trace $t'$, we have $(t' \in \sigma_d(i')) \Rightarrow (t' \in \sigma_o(i'))$. As a result we have $i \xrightarrow{a} i'$ and $t' \in \sigma_o(i')$. By definition of the operational semantics, this implies that $a.t' \in \sigma_o(i)$. Hence the property holds.
\end{itemize}
\end{proof}

\end{document}